\documentclass[twocolumn,prl,twocolumn,superscriptaddress]{revtex4}
\usepackage{color}
\usepackage{amsmath,amsthm,mathrsfs,dsfont}
\usepackage{graphicx}
\usepackage{tikz}
\usepackage{xr}
\usepackage{xc}
\usepackage{float}
\usepackage{bbm}
\usepackage{epsfig} 
\usepackage{verbatim}
\usepackage{array}
\usepackage{setspace}
\usepackage{bbding}
\usepackage{amssymb}
\usepackage{pifont}
\usepackage{braket}

\usepackage{newtxtext,newtxmath}

\newcolumntype{M}[1]{>{\centering\arraybackslash}m{#1}}
\newcolumntype{N}{@{}m{0pt}@{}}

\usepackage[unicode=true,
 bookmarks=true,bookmarksnumbered=false,bookmarksopen=false,
 breaklinks=false,pdfborder={0 0 1},backref=false,colorlinks=true]
 {hyperref}
\hypersetup{
 linkcolor=magenta, urlcolor=blue, citecolor=blue, pdfstartview={FitH}, hyperfootnotes=false, unicode=true}
 
\usepackage{color}
\newtheorem{lemma}{Lemma}

\newtheorem{result}{Result}

\makeatletter
\@ifundefined{textcolor}{}
{%
 \definecolor{BLACK}{gray}{0}
 \definecolor{WHITE}{gray}{1}
 \definecolor{RED}{rgb}{1,0,0}
 \definecolor{GREEN}{rgb}{0,1,0}
 \definecolor{BLUE}{rgb}{0,0,1}
 \definecolor{CYAN}{cmyk}{1,0,0,0}
 \definecolor{MAGENTA}{cmyk}{0,1,0,0}
 \definecolor{YELLOW}{cmyk}{0,0,1,0}
}

\usepackage{amsfonts}\usepackage{tabularx}\usepackage{dcolumn}\usepackage{bm}\usepackage{graphicx}\usepackage{epstopdf}

\hypersetup{urlcolor=blue}

\makeatother

\newcolumntype{C}[1]{>{\centering\arraybackslash$}p{#1}<{$}}

\usepackage{algorithm}
\usepackage{algorithmic}

\begin{document}

\widetext

\title{Low-depth Quantum State Preparation}

\author{Xiao-Ming Zhang}
\affiliation{Department of Physics, City University of Hong Kong, Tat Chee Avenue, Kowloon, Hong Kong SAR, China}

\author{Man-Hong Yung}
\email{yung@sustech.edu.cn}
\affiliation{Department of Physics, Southern University of Science and Technology, Shenzhen, 518055, China}
\affiliation{Shenzhen Institute for Quantum Science and Engineering, Southern University of Science and Technology, Shenzhen, 518055, China}
\affiliation{Guangdong Provincial Key Laboratory of Quantum Science and Engineering, Southern University of Science and Technology, Shenzhen, 518055, China}

\author{Xiao Yuan}
\email{xiaoyuan@pku.edu.cn}
\affiliation{Center on Frontiers of Computing Studies, Department of Computer Science, Peking University, Beijing 100871, China}

\begin{abstract}
A crucial subroutine in quantum computing is to load the classical data of $N$ complex numbers into the amplitude of a superposed $n=\lceil \log_2N\rceil$-qubit state. It has been proven that any algorithm universally implementing this subroutine would need at least $\mathcal O(N)$ constant weight operations. However, the proof assumes that only $n$ qubits are used, whereas the circuit depth could be reduced by extending the space and allowing ancillary qubits. Here we investigate this space-time tradeoff in quantum state preparation with classical data. We propose quantum algorithms with $\mathcal O(n^2)$ circuit depth to encode any $N$ complex numbers using only single-, two-qubit gates and local measurements with ancillary qubits. Different variances of the algorithm are proposed with different space and runtime. In particular, we present a scheme with $\mathcal O(N^2)$ ancillary qubits, $\mathcal O(n^2)$ circuit depth, and $\mathcal O(n^2)$ average runtime, which exponentially improves the conventional bound. 
While the algorithm requires more ancillary qubits, it consists of quantum circuit blocks that only simultaneously act on a constant number of qubits and at most $\mathcal O(n)$ qubits are entangled. 
We also prove a fundamental lower bound $\mit\Omega(n)$ for the minimum circuit depth and runtime with arbitrary number of ancillary qubits, aligning with our scheme with $\mathcal O(n^2)$.
The algorithms are expected to have wide applications in both near-term and universal quantum computing. 
\end{abstract}
\maketitle

Various quantum algorithms have been designed for solving different types of problems\cite{Nielsen.02}. 
A critical subroutine of many quantum algorithms is to encode classical data into a superposed quantum state\cite{Kaye.01,Grover.02,Mottonen.05,Plesch.11,Yung.14,Iten.16,Zhao.19}, which prepares a general muti-qubit state with classically given amplitudes.
An efficient state preparation scheme is the prerequisite of many algorithms, including quantum linear system algorithms\cite{Harrow.09,Wossnig.18}, quantum versions of data fitting\cite{Wiebe.12}, principal component analysis\cite{Lloyd.14}, support vector machine\cite{Rebentrost.14}, Hamiltonian simulation algorithms\cite{Lloyd96,Aharonov.03,Low.17}, quantum machine learning\cite{Biamonte.17,Wan.17,Beer.20,Romero.17,Bondarenko.20,Wang.21,Zhang.21}, etc. 
Theoretically, the minimal number of constant-weight operations to prepare an arbitrary $N$-dimensional or $n=\lceil \log_2N\rceil$ qubit state is lower bounded by $\mit{\Omega}(N/\log n)$\cite{Nielsen.02,note_Nielsen}, which corresponds to the circuit depth of $\mit{\Omega}(N/(n\log n))$. 
For instance, one may construct a unitary to transform $|0\rangle^{\otimes n}$ to the target state with only single-qubit and CNOT gates, and existing algorithms\cite{Mottonen.05,Plesch.11,Zhao.19,Iten.16} require $\mathcal O(N)$ circuit depth, which is close to the fundamental limit. Since the complexity is linear in the dimension $N$ or exponential in the number of qubits $n$, it requires a deep circuit for large $N$. For example, the circuit depth is already challenging for the current technology when $N\approx 10^3$ or $n=10$.

However, the proof of the lower bound considers operations on exactly $n$ qubits, and one may trade the circuit depth (time) with ancillary qubits (space). Along this line, quantum circuits with $\mathcal O(n^2)$ depth have been proposed to encode binary vectors\cite{Cortese.18} and general non-binary vectors\cite{Araujo.20,Johri.20} into special types of entangled states. The key idea is to apply operations on $N$ qubits in parallel so that the circuit depth is polylogarithmic in $N$. In Ref~\cite{Araujo.20}, the method has also been applied for improving quantum machine learning algorithms.
Nevertheless, the output quantum state is encoded with $N$ qubits, which is exponentially larger than $n$ and is in a complicated entangled basis of all the $N$ qubits, which may not be universally usable as the input to other quantum algorithms.
While there are other methods with logarithmic costs, including controlled-rotation-based\cite{Grover.02,Prakash.14}, Grover-oracle-based\cite{Soklakov.06} and quantum random-access-memory-based\cite{Casares.20,Lloyd.13} methods, they require global unitaries or global oracles acting on all qubits, which is challenging based on current technologies.
Therefore, it remains an open question whether it is possible to more efficiently and directly prepare a general $N$-dimensional ($n$-qubit) quantum state with constant-weight operations and a shallow circuit depth.

In this work, we address this problem by introducing probabilistic quantum state preparation algorithms. We consider the task of preparing an $N$-dimensional ($n$-qubit) state and introduce several quantum algorithms that use circuits with polylogarithmic depth $\mathcal O(n^2)$. With different numbers of ancillary qubits, the algorithms have different success probabilities, which could be enhanced to $\mathcal O(1)$ with a runtime that is inverse proportional to the success probability.
As a result, the sequential algorithm uses $\mathcal O(n)$ ancillary qubits with average runtime $\mathcal O(N^2)$ and the parallel algorithms uses more ancillary qubits with a smaller average runtime. Specifically, the extreme parallel algorithm has an average runtime of $\mathcal O(n^2)$ with $\mathcal O(N^2)$ ancillary qubits. {Note that for all the proposed algorithms, one only needs to maintain entanglement of at most $\mathcal O(n)$ qubits.} Our results thus show the space-time trade-off in quantum state preparation. 
{Moreover, we have shown that fundamentally, the circuit depth and runtime is lower bounded by $\mit\Omega(n)$ even with arbitrarily large amount of ancillary qubits, which is comparable to our result of $\mathcal O(n^2)$.  }
\\

\noindent \emph{\textbf{Framework.}}
We first introduce the task of quantum state preparation. Given a vector $\bm{u}:=[u_0,u_1,\cdots,u_{N-1}]\in\mathbb{C}^N$
 of $N$ complex numbers satisfying $\|\bm{u}\|_2=1$, we  consider the preparation of the $n$-qubit state
 \begin{align}
|\psi(\bm{u})\rangle:=\sum_{i=0}^{N-1}u_i|n,i\rangle, \label{eq:amp}
\end{align}
where $|n,i\rangle$ is the $n$-qubit binary representation of $i$. For example, $|3,7\rangle=|111\rangle$, $|3,6\rangle=|110\rangle$ and $|4,7\rangle=|0111\rangle$. Here, the state $|\psi(\bm{u})\rangle$ is also called the \textit{amplitude encoding}\cite{Lloyd.13,Lloyd.18,Dallaire.18,Schuld.19,Schuld.20,Lu.20}  of the vector $\bm{u}$ and  serves as our target state. The problem we consider is as follow:

$\\$
\textit{
Given an arbitrary quantum state described in Eq.~\eqref{eq:amp}, find preparations method with constant-weight operations and polylogarithmic circuit depth.
}
$\\$

To prepare $|\psi(\bm{u})\rangle$, we consider a resized vector $\bm{v}:=\bm{u}/\max(|u_i|)$ and define the \textit{label encoding state} of $\bm{v}$ with $n+1$ qubits  (in below,  quantum states may be represented up to a normalization factor)
\begin{equation}\label{eq:lab}
|\bm{v}\rangle:=\sum_{i=0}^{N-1}|n,i\rangle|v_i\rangle
=\sum_{i=0}^{N-1}|n,i\rangle \left[v_i|0\rangle+(1-v_i)|1\rangle\right], 
\end{equation}
where $|n,i\rangle$ and $|v_i\rangle= v_i|0\rangle+(1-v_i)|1\rangle$ represents the $n$-qubit label and  the value single qubit, respectively. Note that if we project the value qubit to $|0\rangle$ and trace it out, we can probabilistically obtain the target state $|\phi(\bm{u})\rangle$. Thus we first focus on the preparation of the label encoding state. \\

\noindent \emph{\textbf{Positive label state preparation.}}
We first consider the special case with positive amplitudes $\bm{v}\in[0,1]^N$. Our algorithm is based on the following result about concatenating two label encoding states. 

\begin{figure} \centering
\includegraphics[width=\columnwidth]{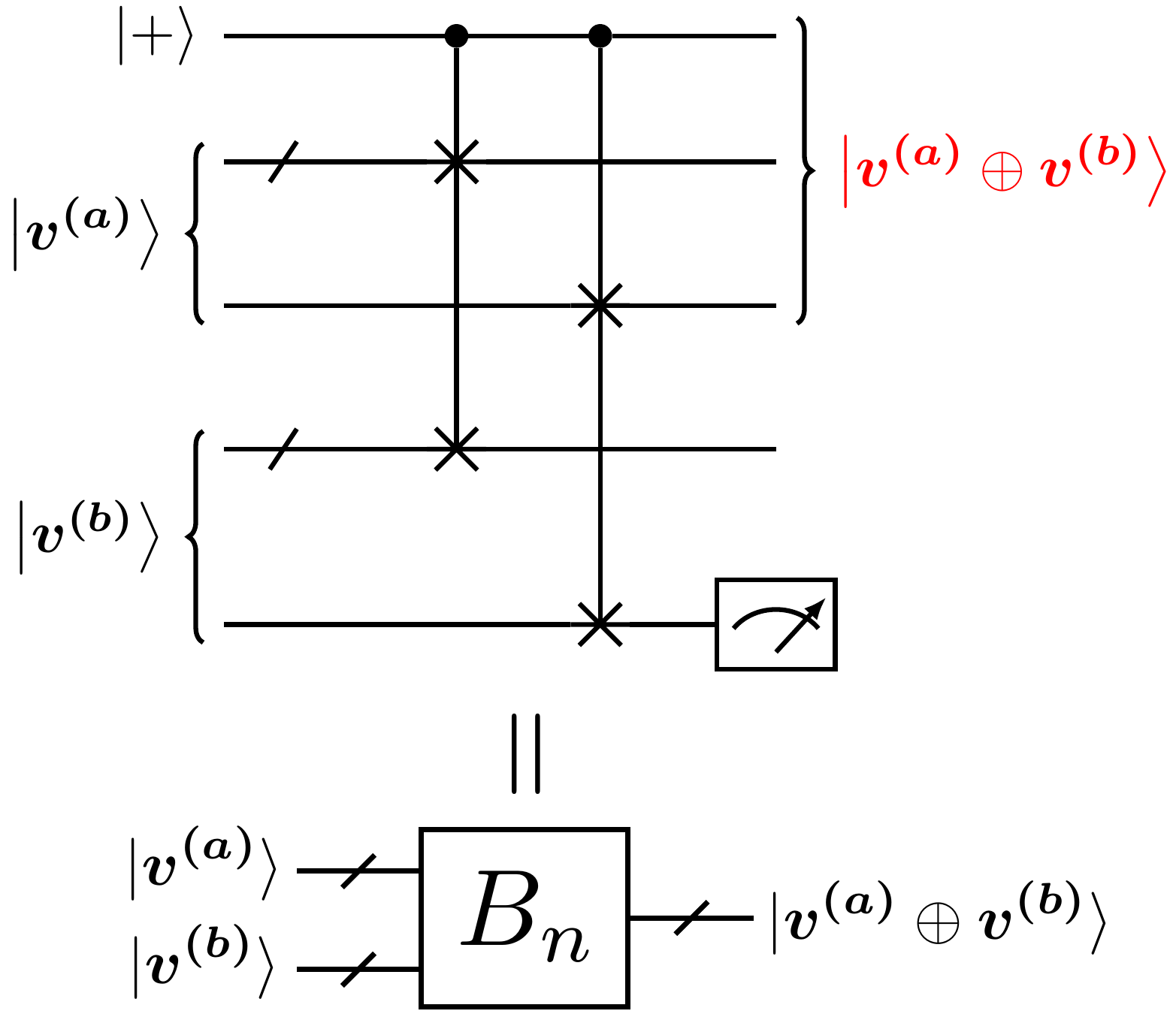}
\caption{Low-depth concatenation circuit for preparing the label encoding state $\big|\bm{v^{(a)}}\oplus\bm{v^{(b)}}\big\rangle$ with input  $|\bm{v^{(a)}}\rangle\otimes|\bm{v^{(b)}}\rangle$ and ancillary qubit $|+\rangle$. Meter represents projecting qubits to state $|+\rangle$, and red colour represents the target output state. See \cite{sm} for details.
}
\label{fig:1}
\end{figure}

\begin{result}\label{th:1}
Given two quantum states $|\bm{v^{(a)}}\rangle$ and $|\bm{v^{(b)}}\rangle$ with $\bm{v^{(a)}}, \bm{v^{(b)}}\in [0,1]^{N}$, there exists an $\mathcal O(n)$ depth concatenation circuit, such that the state $|\bm{v^{(a)}}\oplus\bm{v^{(b)}}\rangle$ can be obtained with probability larger than $1/2$.
\end{result}

\noindent We have defined $\bm{v^{(a)}}\oplus \bm{v^{(b)}}=\left[v^{(a)}_0,\cdots,v^{(a)}_{N-1},v^{(b)}_0,\cdots,v^{(b)}_{N-1}\right]$. The \textit{concatenation circuit} is shown in Fig.~\ref{fig:1}, where we perform a joint controlled-swap operation on each pair of the qubits for $|\bm{v^{(a)}}\rangle$ and $|\bm{v^{(b)}}\rangle$
with a control ancillary qubit initialized in $|+\rangle=1/\sqrt{2}(|0\rangle+|1\rangle)$.
The state then becomes
\begin{equation}
\frac{1}{\sqrt{2}}\left(|0\rangle|\bm{v^{(a)}}\rangle\otimes|\bm{v^{(b)}}\rangle+|1\rangle|\bm{v^{(b)}}\rangle\otimes|\bm{v^{(a)}}\rangle\right),\label{eq:sw}
\end{equation}
where $\otimes$ is the Kronecker product. Next, the key step is to disentangle the last $(n+1)$ qubits by projecting the last value qubit to $|+\rangle$. The success probability of the projection satisfies $p_+\geqslant1/2$ 
(see \cite{sm}). After the projection, the last $n+1$ qubits are disentangled with the rest $n+2$ qubits, and the full quantum state is given by 
\begin{align}
\left(|0\rangle|\bm{v^{(a)}}\rangle+|1\rangle|\bm{v^{(b)}}\rangle\right)\otimes|\bm{v^{\text{uni}}}\rangle\notag = |\bm{v^{(a)}}\oplus \bm{v^{(b)}}\rangle\otimes|\bm{v^{\text{uni}}}\rangle,
\end{align}
where we have defined $\bm{v^{\text{uni}}}=[1/2,1/2,\cdots,1/2]$. Note that $|\bm{v^{(a)}}\oplus \bm{v^{(b)}}\rangle=\sum_{i=0}^{2N-1}|n+1,i\rangle \left[v_i|0\rangle+(1-v_i)|1\rangle\right]$, where $v_i=v_i^{(a)}$ for $i<N$ and $v_i=v_{i-N}^{(a)}$ for $i\geqslant N$. By tracing out $|\bm{v^{\text{uni}}}\rangle$, the label encoding state of the concatenated vectors $|\bm{v^{(a)}}\oplus \bm{v^{(b)}}\rangle$ is obtained. Because $|\bm{v^{(a)}}\rangle$ and $|\bm{v^{(b)}}\rangle$ are $(n+1)$ qubit states, and each control swap gate can be realized with a constant numbers of single and two qubit gates, the concatenation circuit has $\mathcal O(n)$ circuit depth. 

The sequential scheme works by applying the concatenation circuit with a divide-and-conquer strategy (see also Algorithms.~\ref{alg:seq}). As shown in 
Fig.~\ref{fig:2}(a), we sequentially prepare the $(i+1)$-qubit label encoding state via the concatenation  circuit. 
All ancillary qubits are disentangled after the measurement and can be reused, so only $\mathcal O(i)$ ancillary qubits are required at each step. The sequential scheme thus requires $\mathcal O(i)$ ancillary qubits and $\mathcal O(n)$ circuit depth. Meanwhile, we need to repeat the concatenation circuit several times to deterministically obtain the output state. We denote the average runtime for preparing $(i+1)$ qubits positive label encoding states as $T_{\text{pos}}(i)$. We also assume that each layer of quantum gates takes constant operation time. Because the success probability of each concatenation circuit is larger than $1/2$, we have $T_{\text{pos}}(i)\leqslant2[2T_{\text{pos}}(i-1)+k_1i+k_0]$. Here, $k_1$ characterize the runtime for single control swap gate, and $k_0$ characterize the runtime for processes that are independent on $i$, such as detection time and latent time. We show that $T_{\text{pos}}(n)\le \mathcal O(N^2)$.

\begin{algorithm} [H]
\caption{: $f_{\text{seq}}(\bm{x})$  }  
\label{alg:seq}  
\begin{algorithmic}[1]

\STATE \textbf{If} $\bm{x}$ is two-dimensional: 

\STATE \quad prepare $|\bm{x}\rangle$  with the unitary preparation method

\STATE \quad \textbf{Output} $|\bm{x}\rangle$

\STATE \textbf{Else}:

\STATE \quad let $\bm{x^{(a)}}\oplus \bm{x^{(b)}}= \bm{x}$ 

\STATE \quad prepare $|\bm{x^{(a)}}\rangle$ with $f_{\text{seq}}(\bm{x^{(a)}})$ 

\STATE \quad prepare $|\bm{x^{(b)}}\rangle$ with $f_{\text{seq}}(\bm{x^{(b)}})$ 

\STATE \quad perform transformation $|\bm{x^{(a)}}\rangle\otimes|\bm{x^{(b)}}\rangle\rightarrow|\bm{x}\rangle$ (with \\\quad success probability larger that $1/2$)

\STATE \quad \textbf{If} the transformation in line 8 fails:

\STATE \quad\quad go to line  6

\STATE \quad \textbf{Else}: 

\STATE \quad\quad \textbf{Output} $|\bm{x}\rangle$

\end{algorithmic} 
\end{algorithm} 
\begin{algorithm}[H]
\caption{: $f_{\text{para}}(\bm{x},c_0)$ }  
\label{alg:para}  
\begin{algorithmic}[1]

\STATE \textbf{If} $\bm{x}$ is two-dimensional:

\STATE \quad prepare $c_0$ copies of $|\bm{x}\rangle$ with the unitary preparation \\\quad method in parallel

\STATE \quad \textbf{Output} $|\bm{x}\rangle^{\otimes c_0}$

\STATE \textbf{Else}: 

\STATE \quad let $\bm{x^{(a)}}\oplus \bm{x^{(b)}}= \bm{x}$ 

\STATE \quad query $f_{\text{para}}(\bm{x^{(a)}},c_0)$ and $f_{\text{para}}(\bm{x^{(b)}},c_0)$ in parallel, get \\\quad return $|\bm{x^{(a)}}\rangle^{\otimes c_a}$  and $|\bm{x^{(b)}}\rangle^{\otimes c_b}$  

\STATE \quad define $c_{\min}=\min\{c_a,c_b\}$

\STATE \quad perform transformation $|\bm{x^{(a)}}\rangle\otimes|\bm{x^{(b)}}\rangle\rightarrow|\bm{x}\rangle$ for $c_{\min}$ \\\quad times \textit{in parallel}, with $c$ trials success

\STATE \quad\textbf{If} $c=0$:

\STATE \quad\quad go to line 6

\STATE \quad\textbf{Else}:

\STATE \quad\quad \textbf{Output} $|\bm{x}\rangle^{\otimes c}$
\end{algorithmic} 
\end{algorithm} 

The runtime can be improved by parallelization. Firstly, two input states of the concatenation circuit (see Fig.~\ref{fig:1}), can be prepared in parallel. Secondly, one can prepare sufficient copies of the input state, and then perform the concatenation circuit for multiple pairs of input states in parallel. In this way, the success probability of the projection (with at least one successful transformation) will be much higher, and the total runtime could be reduced dramatically. 
In Fig.~\ref{fig:2}(c), we show the \textit{parallel concatenation circuit} $B'_i$ for preparing $(i+1)$-qubit label encoding states, which  receives state $|\bm{v^{(a)}}\rangle^{\otimes c_a}\otimes|\bm{v^{(b)}}\rangle^{\otimes c_b}$. One performs $c_{\min}=\min\{c_a,c_b\}$ times of concatenation circuit in parallel, with totally $c'$ successful trials and obtain the output state $|\bm{v^{(a)}}\oplus\bm{v^{(b)}}\rangle^{\otimes c'}$. Note that $c'$ follows Binomial distribution $c'\sim \text{B}(c_{\min},p_+)$ with $p_+>1/2$.

\begin{figure} [t]\centering
\includegraphics[width=\columnwidth]{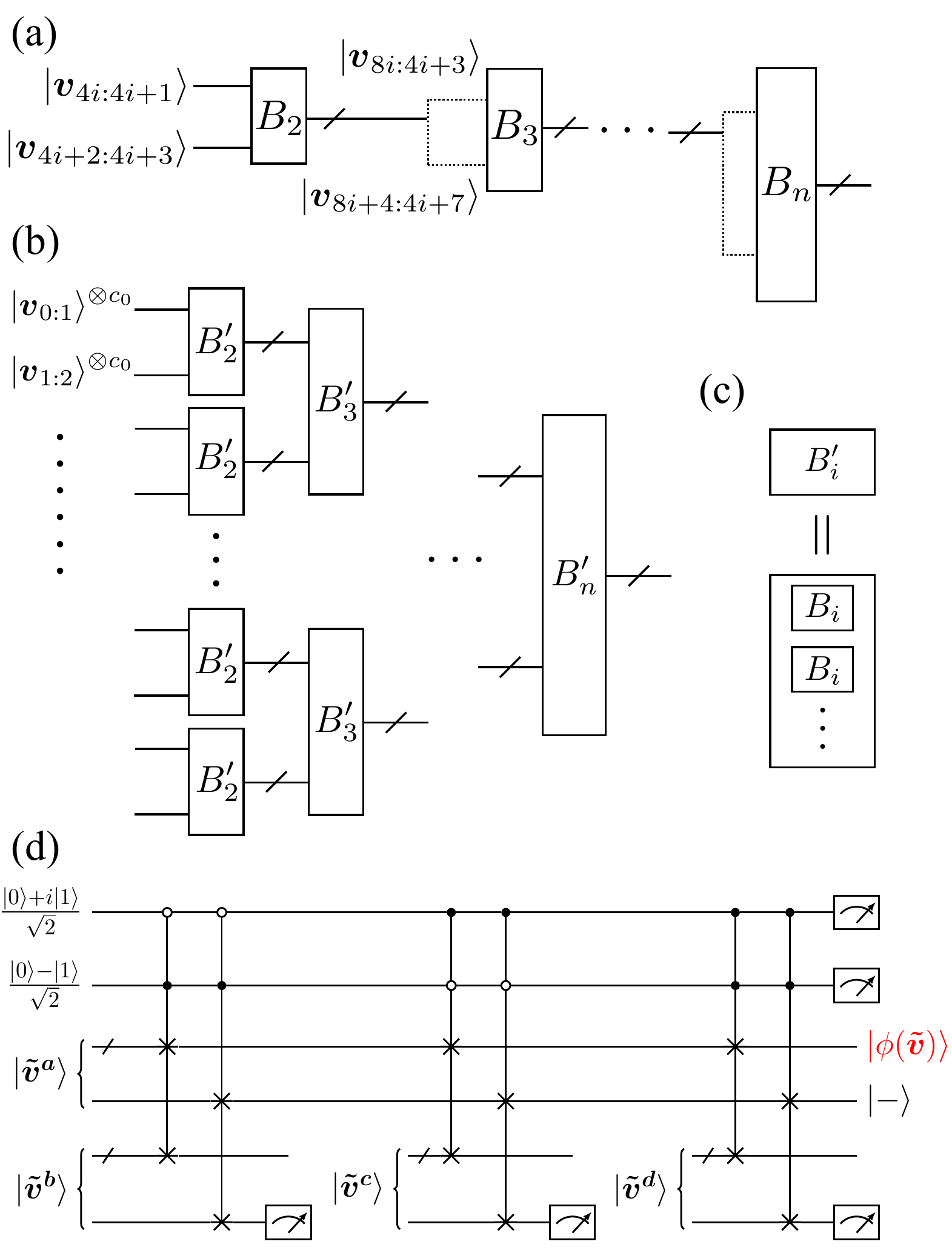}
\caption{(a)-(b) Low-depth quantum circuit for preparing positive label encoding state for (a) sequential parallel preparation and (b) parallel preparation. (c) A parallel concatenation circuit $B_i'$ contains multiple (non-parallel) concatenation circuits $B_i$. (d) Circuit for preparing arbitrary states $|\phi(\bm{\tilde{v}})\rangle$ with four positive label encoding states. Red colour represents the target output state. }
\label{fig:2}
\end{figure}

The parallel preparation of an $(n+1)$-qubit label encoding state is shown in Fig.~\ref{fig:2}(b), where we have denoted $\bm{v}_{i:j}=[v_i,v_{i+1},\dots,v_{j}]$. We prepare $c_0$ copies of the low dimensional label encoding states, i.e. $|\bm{v}_{0:1}\rangle^{\otimes c_0}, |\bm{v}_{1:2}\rangle^{\otimes c_0}, \dots$, which are  concatenated recursively to state $|\bm{v}\rangle$. Whenever the parallel concatenation circuit has zero copies of output, we repeat the preparation of the input state of the corresponding block (see also Algorithm.~\ref{alg:para}). The average runtime and space complexity depends on $c_0$. For example, when $c_0=\left\lceil N+N^{3/4}\right\rceil$, the ancillary qubit number scales as $\mathcal O(N^{2})$, while $T_{\text{pos}}(n)$ scales polylogarithmic as $\mathcal O(n^2)$ (see Sec.~II of \cite{sm}). When $c_0=1$, we need $\mathcal O(N)$ ancillary qubits, and we numerically find that $T_{\text{pos}}(n)=\mathcal O(N^{1.52})$(Sec.~II of \cite{sm}) in the worst case, i.e. $p_+=0.5$. We note that although more ancillary qubits are needed in the parallel schemes, only the entanglement among at most $(n+1)$ qubtis is required, and all ancillary qubits can be reused after the preparation.

With the label encoding state, the amplitude encoding state $|\phi(\bm{v})\rangle$ can be obtained by projecting the value qubit to $|0\rangle$ with probability $p_s$. In this way, an arbitrary quantum state [Eq.~\eqref{eq:amp}] with real amplitudes $u_i$ can be prepared with low circuit depth. We summarize the result as follows. 
\begin{result}\label{Result:positive}
With probability $p_s$, an arbitrary $n$-qubit quantum state with real amplitudes can be prepared via the sequential and parallel algorithms with $\mathcal O(n^2)$ depth of single-qubit gates, two-qubit gates and local measurements.
\end{result}

\noindent We note that the runtime $T$ of deterministically preparing the state is $\mathcal O(T_{\rm pos}/p_s)$, and we will shortly discuss how to bound it in the general case. 

The sequential and parallel algorithms also work for the general case with complex amplitudes. However, the success probability of the concatenation circuit with two arbitrary vectors can only be lower bounded by $p_+\geqslant1/10$, instead of $p_+\geqslant1/2$ for positive amplitudes case.   Therefore the runtime or the number of qubits will be significantly increased~\cite{sm}.\\  

The pseudo code of sequential and parallel preparation methods for preparing a positive quantum label encoding states are provided in \cite{sm}. 

\noindent \emph{\textbf{Arbitrary state preparation.}}
Now we propose an alternative strategy to prepare the label encoding state with arbitrary complex amplitudes. 
We rewrite $\bm{{v}}$ as the combination of four positive vectors, $\bm{{v}}=\bm{ v^a}-\bm{ v^b}+i\bm{ v^c}-i\bm{ v^d}$, whose elements are defined as ${v}^{a}_i=\max\left(\text{Re}({v}_i),0\right), {v}^{b}_i=\max\left(-\text{Re}({v}_i),0\right), {v}^{c}_i=\max\left(\text{Im}({v}_i),0\right),$ and ${v}^{d}_i=\max\left(-\text{Im}({v}_i),0\right)$, respectively. First, the four positive label encoding states $|\bm{ v^a}\rangle$, $|\bm{ v^b}\rangle$, $|\bm{{v}^c}\rangle$ and $|\bm{ v^d}\rangle$ could be prepared with the above scheme. Then, we introduce two ancillary qubits prepared in states $(|0\rangle+i|1\rangle)/\sqrt{2}$ and $(|0\rangle-|1\rangle)/\sqrt{2}$. The entire system is described by 
\begin{equation}
(|00\rangle-|01\rangle+i|10\rangle-i|11\rangle)\otimes|\bm{{v}^{abcd}}\rangle,\label{eq:ini}
\end{equation}
where we have used the abbreviation $|\bm{{v}^{abcd}}\rangle\equiv|\bm{{v}^{a}}\rangle\otimes|\bm{{v}^{b}}\rangle\otimes|\bm{{v}^{c}}\rangle\otimes|\bm{{v}^{d}}\rangle$.
To obtain $\ket{\bm v}$, we perform three sets of controlled-controlled-swap gates, which swap $|\bm{{v}^{a}}\rangle$ and one of the states among $|\bm{{v}^{b}}\rangle$, $|\bm{{v}^{c}}\rangle$ and $|\bm{{v}^{d}}\rangle$, with two ancillary qubits as control qubits. The corresponding quantum circuit is shown in Fig.~\ref{fig:2}~(d), where the hollow nodes represents controlled on $|0\rangle$, and the solid nodes represents controlled on $|1\rangle$.
 The state then becomes
\begin{align}
|00\rangle|\bm{{v}^{abcd}}\rangle-|01\rangle|\bm{{v}^{bacd}}\rangle+i|10\rangle|\bm{{v}^{cbad}}\rangle-i|11\rangle|\bm{{v}^{dbca}}\rangle.\label{eq:cc}
\end{align}
The above operations require totally $3(n+1)$ control-control-swap gates, where each of which can be decomposed to constant numbers of single- and two-qubit gates, so the corresponding circuit depth is $\mathcal O(n)$.

In the next step, we project two ancillary qubits and the label qubits of the last three label encoding states to $|+\rangle$. If the projection succeeds (the success probability {$p'_s$} will be discussed later), the quantum state becomes
\begin{align}
|+\rangle^{\otimes2}(|\bm{{v}^{a}}\rangle-|\bm{{v}^{b}}\rangle+i|\bm{{v}^{c}}\rangle-i|\bm{{v}^{d}}\rangle)\otimes|\bm{v^{\text{uni}}}\rangle^{\otimes3}. \label{eq:abcd}
\end{align}
Since $|\bm{{v}^{a}}\rangle-|\bm{{v}^{b}}\rangle+i|\bm{{v}^{c}}\rangle-i|\bm{{v}^{d}}\rangle=\sqrt{2}|\psi(\bm{{v}})\rangle\otimes|-\rangle$, we can  trace out $|+\rangle^{\otimes2}$ and $|-\rangle\otimes|\bm{v^{\text{uni}}}\rangle^{\otimes3}$ to have the target state $|\psi(\bm{{v}})\rangle$. 
Together  with  Result~\ref{Result:positive}, we have the following result.
\begin{result}
With probability $p_s'$, an arbitrary $n$-qubit quantum state can be prepared via the sequential and parallel algorithms with $\mathcal O(n^2)$ depth of single-qubit gates, two-qubit gates and local measurements.
\end{result}
\noindent We note that the average runtime is proportional to the $T_{\text{pos}}$ for preparing each $|\bm{ v^a}\rangle$, $|\bm{ v^b}\rangle$, $|\bm{{v}^c}\rangle$ or $|\bm{ v^d}\rangle$, divided by the projection success probability $p_s'$, i.e., $\mathcal O(T_{\text{pos}}/p_s')$. Next we show how to estimate the success probabilities.  \\

\noindent \emph{\textbf{Projection success probability.}}
To exactly prepare the amplitude encoding state $\ket{\psi(\bm u)}$, the projection probabilities $p_s$ (for positive data) and $p_s'$ (for complex data) are both lower bounded by $\mit \Omega(\sum_i |u_i|^2/\max(|u_i|^2)N)$.  The worst case lower bound is $\mit \Omega(1/N)$ and it could be tightened with a detailed analysis. Denote $u_i = |a_i|/\sqrt{\sum_i |a_i|^2}$ for positive data or $u_i= a_i/\sqrt{\sum_i |a_i|^2}$ for general complex data, we consider that the classical data $\bm u$ is randomly generated in two ways. 
\begin{enumerate}
	\item For the first way, we let $a_i = b_ie^{i\phi_i}$, and uniformly generate each $b_i$ from $[-1,1]$, and $\phi_i$ from $[0,\pi]$.
	{\item For the second way, we let $a_i = a^\text{(r)}_i+ia^\text{(m)}_i$, and 
generate each $a^\text{(r)}_i,a^\text{(m)}_i$ according to the standard normal distribution $\mathcal N(0,1)$.}
\end{enumerate}
\noindent The first way corresponds to the case where the classical data, i.e., each $u_i$, is uniformly random; And the second way corresponds to the case where the state vector $\ket{\psi(\bm u)}$ is uniformly random in the Hilbert space\cite{Watrous.18}. Then we can lower bound the projection probability from the Chernoff bound as follows (see Sec.III A of \cite{sm}).
\begin{result}\label{result:perfect}
With failure probability $\delta \in (0,1)$, the projection probabilities are lower bounded by
\begin{equation}
\begin{aligned}
	\textrm{case 1}:\quad	&p_s, p_s'\ge \mit\Omega\left(\delta^{1/N}\right),\\
	\textrm{case 2}:\quad	&p_s, p_s' \ge \mit\Omega\left(\delta^{1/N}/\log (N/\delta)\right).
\end{aligned}
\end{equation}
\end{result}
\noindent Therefore, the projection probabilities could be lower bounded by a constant for case 1 and by $1/n$ for case 2 given sufficiently large $N$. 

We can further improve the projection probabilities for case 2 by allowing approximate   state preparation. We introduce a cut-off value $u_{\text{cut}}$ and define $\tilde{v}_i\equiv\arg(u_i)\min(|u_i|/u_\text{cut},1)$. After preparing $|\phi(\bm{\tilde{v}})\rangle$ with $\bm{\tilde{v}}=[\tilde{v}_0,\tilde{v}_1,\cdots,\tilde{v}_{N-1}]$, we can achieve the preparation fidelity $F\equiv|\langle\phi(\bm{\tilde{v}})|\phi(\bm{u})\rangle|^2\geqslant1-\varepsilon_{\text{th}}$ by appropriately choosing the cutoff $u_{\text{cut}}$. Note that a perfect preparation $F=1$ corresponds to $u_{\text{cut}}=\max(|u_i|)$.
In the worst cases, we have $p_s\geqslant\text{mean}(|\tilde{v}_i|^2) $ and $p'_s\geqslant \text{mean}(|\tilde{v}_i|^2)/64$ (see Sec.III B of \cite{sm}). 
These values decrease with $u_{\text{cut}}$, whereas the preparation fidelity $F$ increases with $u_{\text{cut}}$, indicating a trade-off between projection probability and fidelity (Sec.III B of \cite{sm}). By setting $u_\text{cut}$ appropriately, $p_s$ ($p_s'$) has logarithmic relation with $\varepsilon_{\text{th}}$ and $\delta$.
We summarize our results as follows.
\begin{result}\label{result:cut}
With threshold infidelity $\varepsilon_{th}\in(0,1)$ and failure probability $\delta\in(0,1)$, the projection probabilities of obtaining the final quantum states with fidelities $F\geqslant1-\varepsilon_{th}$ are lower bounded by
\begin{equation}
\begin{aligned}
	&\textrm{case 2}:\quad	p_s, p_s' \ge \mit\Omega\left(\frac{\delta^{2/N}}{\log\delta^{-1}\log(\delta^{-1}\varepsilon_{\text{th}}^{-1})}\right).\label{eq:cut}
\end{aligned}
\end{equation}
\end{result}
Recall that the total runtime $T$ scales as $\mathcal O(T_{\text{pos}}/p_s)$ or $\mathcal O(T_{\text{pos}}/p_s')$. Combining the above results, we summarize the circuit depth, runtime, number of qubits in Table~\ref{tab:compare}. \\

 \begin{figure} [!t]\centering
\includegraphics[width=0.9\columnwidth]{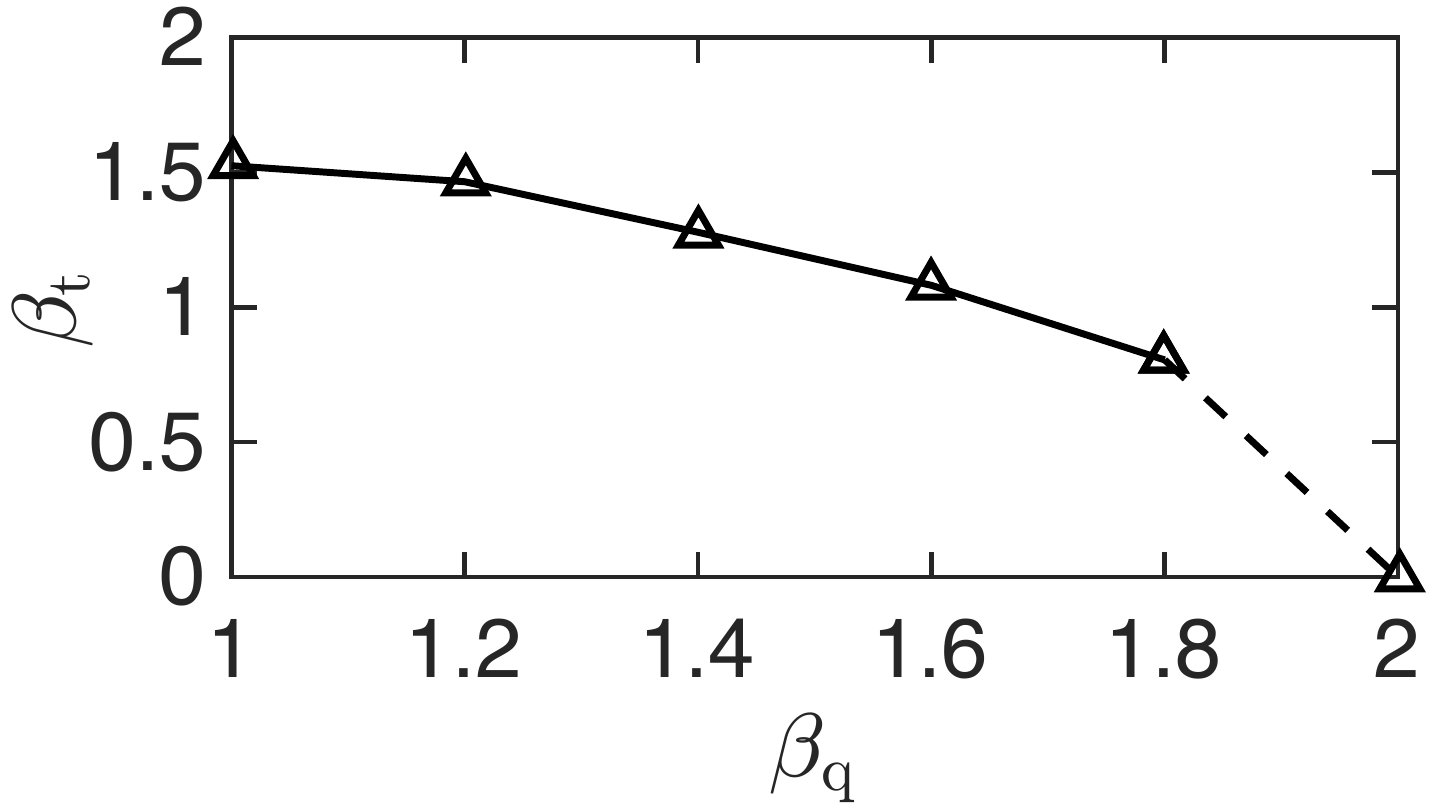}
\caption{{Space-time tradeoff in state preparation. The relation between time $\left\lceil N^{\beta_t}\right\rceil$ with exponent $\beta_t$ and space $\mathcal O(N^{\beta_q})$ with exponent $\beta_q$ for different parallel preparation schemes. When $\beta_q=2$, $T$ no longer follow polynomial scaling.  \label{fig:tf}}
}
\end{figure}

\noindent \emph{\textbf{Space-time trade-off.}}
We note that all the three listed algorithms have circuit depth $\mathcal O(n^2)$ with different runtime and number of qubits. We can see that the runtime decreases with more qubits, indicating a space-time tradeoff in quantum state preparation. { Furthermore, we change $c_0$ for parallel preparation such that there are totally $\left\lceil N^{\beta_t}\right\rceil$ qubits with $1\leqslant\beta_q<2$, and assume the average runtime is in the form of $T=O(N^{\beta_q})$. We numerically estimate the exponents $\beta_t$ for different $\beta_q$. As shown in Fig.~\ref{fig:tf}, the exponents $\beta_t$ decreases rapidly with larger $\beta_q$. Note that when $\beta_q$, $T$ does not follow the polynomial scaling, which is consistent with our analytical estimation.}

Moreover,  we proof a fundamental lower bound of both the circuit depth and runtime for quantum state preparation (Sec.~V of \cite{sm}). If only constant-weight operations are allowed and $\bm{u}$ is stored classically, we have the following result.
\begin{result}\label{th:bnd}
Preparing an arbitrary quantum state from classical amplitudes $\bm{u}$ requires at least circuit depth and runtime $\mit\Omega(n)$.
\end{result}
\noindent The basic idea is that when $\bm{u}$ is stored classically, one requires at least $2^n$ bits to store it. In order to ``compress'' the information spread over $2^n$ bits, at least $n$ layers of the local operations circuit are required. This lower bound is independent on the number of ancillary qubits, and measurement and post-selections are also taken into consideration. Our work gives an explicit construction with circuit depth and runtime $\mathcal O(n^2)$, which is comparable to the lower bound. \\

\noindent \emph{\textbf{Discussion.}}
 We have demonstrated several protocols to prepare an arbitrary $N$-dimensional quantum state with $\mathcal O((\log_2(N))^2)$ circuit depth, different number of ancillary qubits, and different runtime. A comparison of our methods to existed methods has been summarized in Table~\ref{tab:compare}. We also discuss the space-time trade-off of the parallel preparation. In Sec.IV of Ref~\cite{sm}, we further discuss the tradeoff between our state preparation methods and existing ones.
 
Besides the low-depth nature, there are other advantages of our work. First, our methods only require to simultaneously maintain entangled states of at most $\mathcal O(n)$ qubits and the rest ancillary qubits are prepared in a separate state. Second,
our methods have a much weaker requirement on the circuit programmability, since most parts of the circuit are fixed except for the first few layers. Moreover, our methods do not require heavy classical computation to compile the circuit, which typically takes time of $\mathcal O(N)$ for unitary state preparation.

\begin{table}[t]\centering
\caption{\label{tab:compare} Comparison of  different state preparation methods. Depth~-~circuit depth; Runtime - circuit runtime $\times$ repetitions; Qubits - total number of qubits; Parallel-1 and -2 corresponds to parallel preparation with $c_0=\left\lceil N+N^{3/4}\right\rceil$ and $c_0=1$. The average runtime for parallel-2 method is estimated with numerical simulation.}
\begin{ruledtabular}
\begin{tabular}{cccccc}

&Depth &Runtime&Qubits\\ 
\hline
Unitary\cite{Mottonen.05,Plesch.11}    & $\mathcal O\big(N\big)$         &$\mathcal O\big(N\big)$     &$\mathcal O\big(n\big)$  \\
Sequential                               & $\mathcal O\big(n^2\big)$ &$\mathcal O\big(N^2\big)$ &$\mathcal O\big(n\big)$ \\
Parallel-1
&$\mathcal O\big(n^2\big)$   &$\mathcal O\big(n^2\big)$     &$\mathcal O\big(N^{2}\big)$\\
Parallel-2                               &$\mathcal O\big(n^2\big)$   &$\mathcal O\big(N^{1.52}\big)$     &$\mathcal O\big(N\big)$ \\

\end{tabular}
\end{ruledtabular}
\end{table}

There are several open questions to be addressed. First, the preparation time could be much longer in the worse case. For example, if $\bm{v}$ is a sparse vector with only a constant number of nonzero elements and bounded values, the success rate $p'_s$ of the projection in Eq.~\eqref{eq:cc} and Eq.~\eqref{eq:abcd} decreases linearly with $N$ and the total runtime will be $N$ times larger. When having a too small success rate, an interesting future work is to design alternative state preparation methods that exploit the sparsity and the structure of the amplitudes.  
Second, it is currently unclear if our methods are optimal in terms of either circuit depth or runtime. {While a lower bound has been derived in Results~\ref{th:bnd}, closing the gap between $\mit\Omega(n)$ and $\mathcal O(n^2)$ would be compelling for both theoretical and practical purposes.} Finally, it is interesting to investigate applications of our methods in existing quantum algorithms and study their performance with noisy intermediate-scaled quantum hardware\cite{Preskill.18,Arute.19,Google.20,Gong.21,cerezo2020variational,endo2021hybrid}. The quantum advantage\cite{Google.20,Zhong.20,Yung.19,Wu.20,Chen.21} could be expected with a robust, efficient and general quantum state preparation protocol. 

\noindent\textbf{Acknowledgements} We thank Tianyang Tao for helpful discussion. This work is supported by the Open Project of Shenzhen Institute of Quantum Science and Engineering (Grant No.~SIQSE202008), Natural Science Foundation of Guangdong Province(Grant No.2017B030308003), the Key R$\&$D Program of Guangdong province (Grant No. 2018B030326001), the Science,Technology and Innovation Commission of Shenzhen Municipality (Grant No.JCYJ20170412152620376 and No.JCYJ20170817105046702 and  $\\$
No.KYTDPT20181011104202253), National Natural Science Foundation of China (Grant No.11875160 and No.U1801661), the Economy, Trade and Information Commission of Shenzhen Municipality (Grant No.201901161512), and Guangdong Provincial Key Laboratory (Grant No.2019B121203002).
%

\vspace{1cm}
\newpage
\onecolumngrid

\begin{center}
{\bf\large Supplemental Material}
\end{center}
\vspace{0.5cm}

\setcounter{secnumdepth}{3}  
\setcounter{equation}{0}
\setcounter{figure}{0}
\setcounter{table}{0}
\renewcommand{\theequation}{S-\arabic{equation}}
\renewcommand{\thefigure}{S\arabic{figure}}
\renewcommand{\thetable}{S-\Roman{table}}
\renewcommand\figurename{Supplementary Figure}
\renewcommand\tablename{Supplementary Table}

\newcolumntype{M}[1]{>{\centering\arraybackslash}m{#1}}
\newcolumntype{N}{@{}m{0pt}@{}}

\makeatletter \renewcommand\@biblabel[1]{[S#1]} \makeatother

\onecolumngrid

\section{positive label encoding states}

Here, we provide more details on the concatenation circuit.  We will provide the proof of Result~\ref{th:1}, and discuss the success probability for the general complex vectors. 

Initially, we are given the product of an extra qubit at state $|+\rangle = (\ket{0}+\ket{1})/\sqrt{2}$ and the label encoding states of $\bm{v^{(a)}}$ and $\bm{v^{(b)}}$. The initial state is given by (up to a normalization factor):
\begin{align}
|+\rangle|\bm{v^{(a)}}\rangle|\bm{v^{(b)}}\rangle.
\label{eq:s1}
\end{align}
According to definition, $\bm{v^{(a)}}, \bm{v^{(b)}}$ satisfies $v_i^{(a)},v_i^{(b)}\in \mathbb{C}$ and $|v_i^{(a)}|,|v_i^{(b)}|\leqslant1$.
Our goal is to transform Eq.~\eqref{eq:s1} to $\big|\bm{v^{(a)}}\oplus\bm{v^{(b)}}\big\rangle=|0\rangle\big|\bm{v^{(a)}}\rangle+|1\rangle|\bm{v^{(b)}}\big\rangle$. As described in the main text, we first apply a set of controlled-swap gates on $|\bm{v^{(a)}}\rangle$ and $|\bm{v^{(b)}}\rangle$ with the extra qubit as the control gate [see Fig.~\ref{fig:swap}]. The (unormalized) quantum state then becomes
\begin{align}
|\Psi_{\text{0}}\rangle=&\frac{1}{\sqrt{2}}|0\rangle|\bm{v^{(a)}}\rangle\oplus|\bm{v^{(b)}}\rangle+\frac{1}{\sqrt{2}}|1\rangle|\bm{v^{(b)}}\rangle\oplus|\bm{v^{(a)}}\rangle\notag\\
=&\frac{1}{\sqrt{2}}|0\rangle|\bm{v^{(a)}}\rangle\sum_{i=0}^{N-1}|n,i\rangle|v^{(b)}_i\rangle +\frac{1}{\sqrt{2}}|1\rangle|\bm{v^{(b)}}\rangle\sum_{i=0}^{N-1}|n,i\rangle|v^{(a)}_i\rangle. \label{eq:sup_swp}
\end{align}
We then project the last qubit (at state $|v^{(a)}_i\rangle$ or $|v^{(b)}_i\rangle$) to state $|+\rangle$. This is equal to applying the projection operator $M=\mathbb{I}_{2^{2n+2}}\otimes|+\rangle\langle+|$ to Eq.~\eqref{eq:sup_swp}. Because $\langle+|v^{(a)}_i\rangle=\langle+|v^{(b)}_i\rangle=1/\sqrt{2}$, we have  
\begin{align}
|\Psi_{\text{1}}\rangle=&M|\Psi_{\text{0}}\rangle\notag\\
=&M\left(\frac{1}{\sqrt{2}}|0\rangle|\bm{v^{(a)}}\rangle\sum_{i=0}^{N-1}|n,i\rangle|v^{(b)}_i\rangle +\frac{1}{\sqrt{2}}|1\rangle|\bm{v^{(b)}}\rangle\sum_{i=0}^{N-1}|n,i\rangle|v^{(a)}_i\rangle\right)\notag\\
=&\frac{1}{2}|0\rangle|\bm{v^{(a)}}\rangle\sum_{i=0}^{N-1}|n,i\rangle|+\rangle +\frac{1}{2}|1\rangle|\bm{v^{(b)}}\rangle\sum_{i=0}^{N-1}|n,i\rangle|+\rangle\notag\\
=&\frac{\sqrt{2}}{2}\left(|0\rangle\big|\bm{v^{(a)}}\big\rangle +|1\rangle\big|\bm{v^{(b)}}\big\rangle\right)\otimes|\bm{v_{\text{uni}}}\rangle\notag\\
=&\frac{1}{\sqrt{2}}\big|\bm{v^{(a)}}\oplus\bm{v^{(b)}}\big\rangle \otimes|\bm{v_{\text{uni}}}\rangle.\label{eq:sup_prj}
\end{align}

The normalization factor in Eq.~\eqref{eq:sup_swp} and Eq.~\eqref{eq:sup_prj} are 
\begin{subequations}
\begin{align}
\langle\Psi_{\text{0}}|\Psi_{\text{0}}\rangle=&(\sum_{i=0}^{N-1}|v_i^{(a)}|^2+|1-v_i^{(a)}|^2)(\sum_{i=0}^{N-1}|v_i^{(b)}|^2+|1-v_i^{(b)}|^2)=A_aA_b,\\
\langle\Psi_{\text{1}}|\Psi_{\text{1}}\rangle=&\frac{1}{2}\sum_{i=0}^{N-1}\left(|v_i^{(a)}|^2+|1-v_i^{(a)}|^2+|v_i^{(b)}|^2+|1-v_i^{(b)}|^2\right) N\left(\frac{1}{2}\right)^2\times2\notag\\
=&\frac{N}{4}(A_a+A_b),
\end{align}
\end{subequations}
\begin{figure} [t]\centering
\includegraphics[width=0.4\columnwidth]{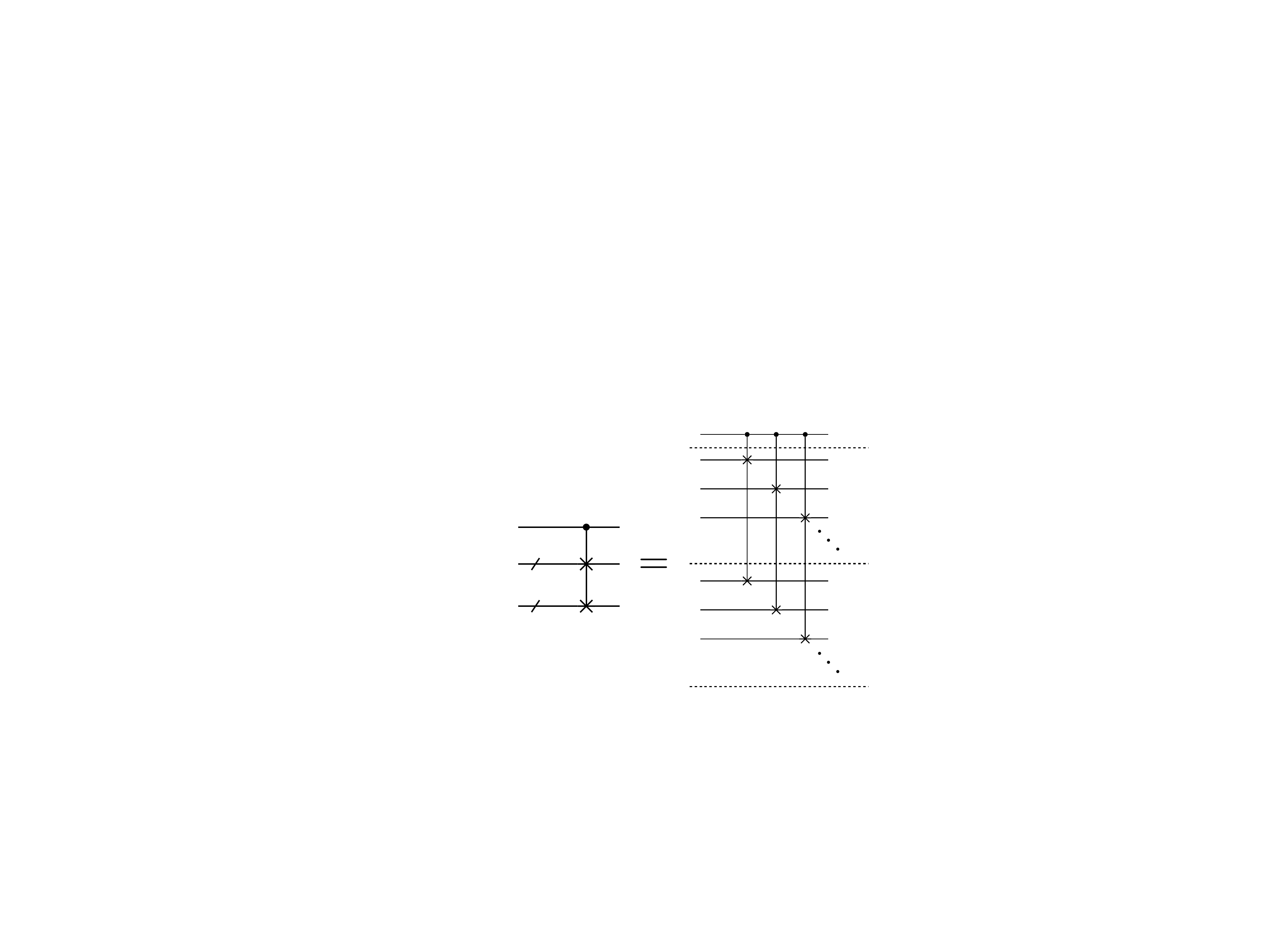}
\caption{Abbreviation of the control swap gate set.}
\label{fig:swap}
\end{figure}
where we have defined $A_a=\sum_{i=0}^{N-1}|v_i^{(a)}|^2+|1-v_i^{(a)}|^2$ and $A_b=\sum_{i=0}^{N-1}|v_i^{(b)}|^2+|1-v_i^{(b)}|^2$. 
The success probability of the projection can be calculated with the normalization factor of the quantum states before and after the projection
\begin{align}
p_+=\frac{|\langle\Psi_{\text{1}}|\Psi_{\text{1}}\rangle|}{|\langle\Psi_{\text{0}}|\Psi_{\text{0}}\rangle|}=\frac{N}{4}\frac{A_a+A_b}{A_aA_b}.\label{eq:+}
\end{align}
In Result~\ref{th:1}, all amplitudes are assumed to be positive, i.e. $v^{(a,b)}_i\in[0,1]$. So we have $A_{a,b}\in[\frac{1}{2}N,N]$. According to Eq.~\eqref{eq:+}, the success probability satisfies 
\begin{align}
p_+\in\left[\frac{1}{2},1\right].
\end{align}
So Result~\ref{th:1} holds true.

If $v_i^{(a,b)}$ are complex values, we have $A_{a,b}\in[\frac{1}{2}N,5N]$, which gives
\begin{align}
p_+\in\left[\frac{1}{10},1\right].
\end{align}
So $p_+$ is lower bounded by $1/10$.

\section{Runtime for parallel preparation of positive label encoding state}\label{sec:3}
In parallel quantum state preparation (Algorithm.~2 
in the main text), $c_0$ determines how many copies of each two-qubit state we should prepare. With larger $c_0$, it is less likely to have $c=0$ at line $9$, and therefore the average runtime is lower. In Sec.~\ref{sec:poly_c_0}, we estimate the average runtime for $c_0=\left\lceil N^{\beta_q-1}\right\rceil$ ($1\leqslant\beta_q<2$) numerically; in Sec.~\ref{sec:N_c_0}, we proof that when $c_0=\left\lceil K(N+N^{3/4})\right\rceil $, the average runtime is $\mathcal O\left((\log N)^2\right)$.

\subsection{$c_0=\left\lceil N^{\beta_q-1}\right\rceil $}\label{sec:poly_c_0}
The average runtimes are estimated numerically by simulating Algorithm.~2
. At each run, we initialize the number of steps as $t_{\text{stp}}=0$. Whenever line $8$ is reached, we update $t_{\text{stp}}=t_{\text{stp}}+\log_2\dim(\bm{x})$. This is because the runtime for line $8$ dominates the total runtime, and its circuit depth is proportional to $\log_2\dim(\bm{x})$. Algorithm.~2
 is run for $1000$ times for $N\in[4,10]$. We assume that the total average runtime increases polynomially as $\mathcal O(N^{\beta_t})$. For each $\beta_q$, the time exponent $\beta_t$ is estimated according to the slop of plot $\log_2(\overline{t_{\text{stp}}})$ vs $\log_2(N)$. Here, $\overline{t_{\text{stp}}}$ is the mean of $t_{\text{stp}}$ at the final step. For example, in Fig.~\ref{fig:s2}, the slop for $c_0=1$ is $1.52$, so the average runtime $T_{\text{pos}}$ can be estimated as $\mathcal O(N^{1.52})$.

\begin{figure} [!htbp]
\includegraphics[width=0.45\columnwidth]{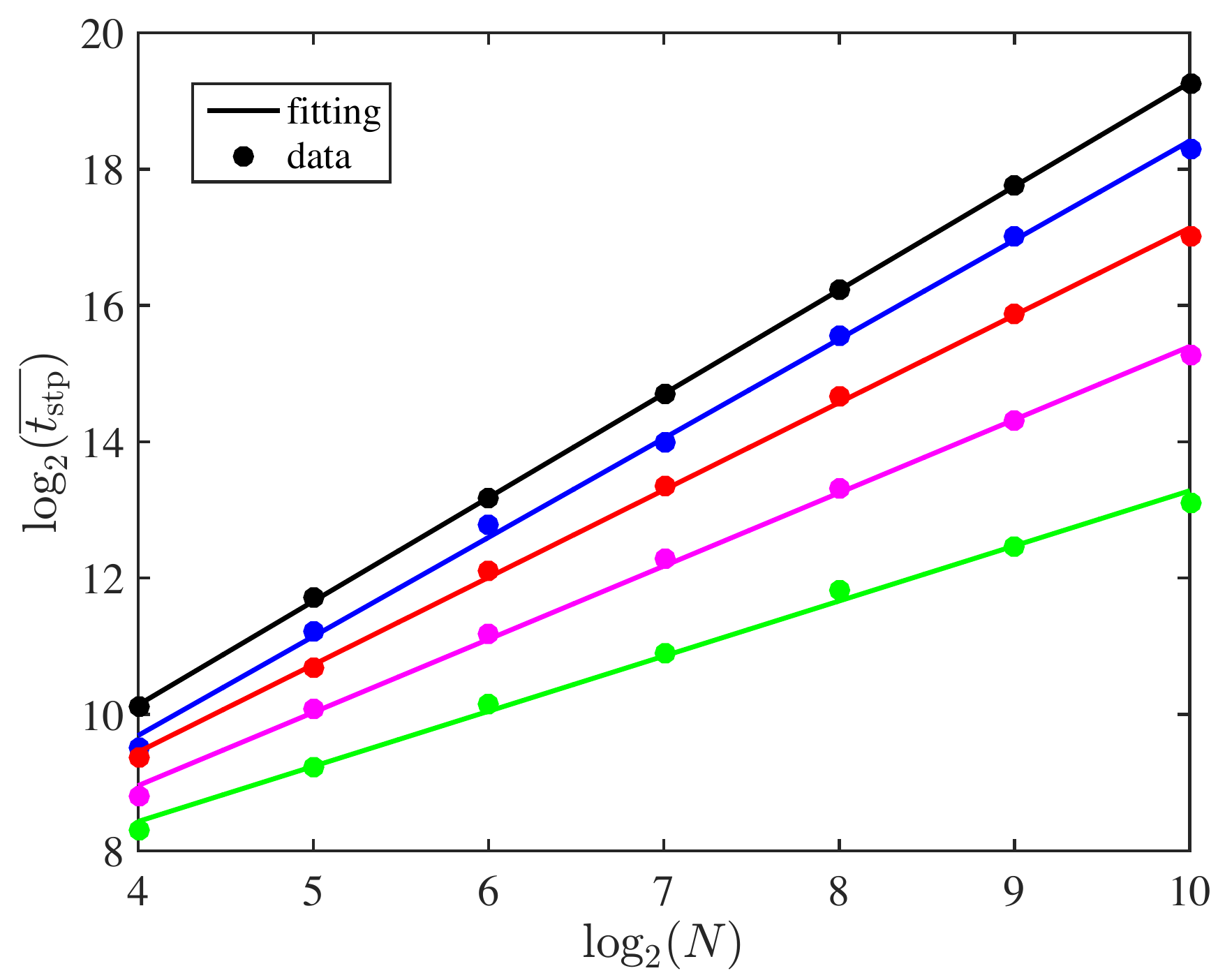}
\caption{ Numerical results of the average runtime for Algorithm.~2
 for $c_0=1$ (black), $c_0=\lceil N^{0.2}\rceil$ (blue), $c_0=\lceil N^{0.4} \rceil$ (red), $c_0=\lceil N^{0.6}\rceil$ (megenta) and $c_0=\lceil N^{0.8}\rceil$ (green). Dots are simulation results of $t_{\text{stp}}$ averaged over $1000$ times, lines are the corresponding linear fittings.}
\label{fig:s2}
\end{figure}

\begin{algorithm} [t]
\caption{: $\hat{g}_{\text{para}}(\bm{x},c_0)$  (subrutine of Algorithm~\ref{alg:4})}  
\label{alg:3}  
\begin{algorithmic}[1]

\STATE \textbf{If} $\bm{x}$ is two-dimensional:

\STATE \quad prepare $c_0$ copies of $|\bm{x}\rangle$ with unitary preparation method \textit{in parallel}

\STATE \quad \textbf{Output} $|\bm{x}\rangle^{\otimes c_0}$

\STATE \textbf{Else}: 

\STATE \quad let $i=\log_2[\dim(\bm{x})]-1$ 

\STATE \quad let $\bm{x^{(a)}}\oplus \bm{x^{(b)}}= \bm{x}$ 

\STATE \quad query $g_{\text{para}}(\bm{x^{(a)}},c_0)$ and $g_{\text{para}}(\bm{x^{(b)}},c_0)$ \textit{in parallel}, get return $|\bm{x^{(a)}}\rangle^{\otimes c_a}$  and $|\bm{x^{(b)}}\rangle^{\otimes c_b}$  

\STATE \quad define $c_{\min}=\min\{c_a,c_b\}$

\STATE \quad perform transformation $|\bm{x^{(a)}}\rangle\otimes|\bm{x^{(b)}}\rangle\rightarrow|\bm{x}\rangle$ for $c_{\min}$ times \textit{in parallel}, with $c'(i)$ trials success

\STATE \quad \textbf{Output} $|\bm{x}\rangle^{\otimes c'(i)}$

\end{algorithmic} 
\end{algorithm} 
\begin{algorithm}[!t]
\caption{: $g_{\text{para}}(\bm{x},c_0)$  }  
\label{alg:4}  
\begin{algorithmic}[1]

\STATE query $\hat{g}_{\text{para}}(\bm{x},c_0)$ to obtain $|\bm{x}\rangle^{\otimes c}$.

\STATE \textbf{If} $c>0$:

\STATE \quad \textbf{Output} $|\bm{x}\rangle^{\otimes c}$

\STATE \textbf{Else if} $c=0$:

\STATE \quad go to line 1
\end{algorithmic} 
\end{algorithm} 

\subsection{$c_0=\left\lceil N+N^{3/4}\right\rceil $} \label{sec:N_c_0}
In the following, we show that by setting $c_0=\left\lceil N+N^{3/4}\right\rceil$, the parallel preparation with Algorithm.~2 has an average runtime of $\mathcal O(n^2)$. To facilitate the discussion, we introduce a variance of  parallel preparation Algorithm~\ref{alg:4}. Obviously, the average runtime of Algorithm~\ref{alg:4} is always lower than Algorithm~2. So we can just focus on Algorithm~\ref{alg:4} in the following, and show that it has average runtime of $\mathcal O(n^2)$.

For a $2^{n+1}$ dimensional input vector $\bm{x}$, the average runtime of Algorithm~\ref{alg:3} is $\mathcal O(n^2)$. If one can show that, at line $\#$10, $\text{Pr}\left[c'(n)>0\right]$ (the probability that $c'(n)>0$) is larger than a non-zeros constant (for arbitrarily large $n$), the average runtime of Algorithm~\ref{alg:4} can be bounded by $\mathcal O(n^2)$.

To facilitate the discussion, we define $c_\text{bnd}(n,i)\equiv 2^{n-i}+2^{3(n-i)/4}$ (note that $c_\text{bnd}(n,i)>0$), and denote $P_{n,i}\equiv\text{Pr}[c'(i)>c_\text{bnd}(n,i)]$, $P_{n}\equiv P_{n,n}$. Now, one just need to show that $P_n$ is alway lower bounded by a non-zero constant.

At line $7$ of Algorithm~\ref{alg:3}, the probability of both $c_a>c_\text{bnd}(n,i-1)$ and  $c_b>c_\text{bnd}(n,i-1)$ are $P_{n,i-1}$. So for $c_{\min}$ at line $8$, we have $\text{Pr}[c_{\min}>c_\text{bnd}(n,i-1)]=P_{n,i-1}^2$. Therefore
\begin{align}
P_{n,i}&\geqslant\text{Pr}[c'(i)>c_\text{bnd}(n,i)|c_{\min}>c_\text{bnd}(n,i-1)]\cdot\text{Pr}[c_{\min}>c_\text{bnd}(n,i-1)]\notag\\
&\geqslant \text{Pr}[c'(i)>c_\text{bnd}(n,i)|c_{\min}>c_\text{bnd}(n,i-1)]P_{n,i-1}^2\notag\\
&\geqslant \text{Pr}[c'(i)>c_\text{bnd}(n,i)|c_{\min}=c_\text{bnd}(n,i-1)]P_{n,i-1}^2. \label{eq:ncc}
\end{align}
The last inequality is because $c'(i)$ follows binomial distribution $c'(i)\sim\text{B}(c_{\min},p_+)$ with $p_+>1/2$, and the probability that $c'(i)>c_\text{bnd}(n,i)$ increases monotonically with $c_{\min}$. $\text{Pr}[c'(i)>c_\text{bnd}(n,i)|c_{\min}=c_\text{bnd}(n,i-1)]$ is just the cumulative distribution function of Binomial distribution, and according to Hoeffding's inequality, we have 
\begin{align}
\text{Pr}[c'(i)>c_\text{bnd}(n,i)|c_{\min}=c_{\text{bnd}}(n,i-1)]&>f(n,i), \label{eq:hoe}
\end{align}
where
\begin{align}
f(n,i)\equiv1-\exp\left(-2c_\text{bnd}(n,i-1) \left(\frac{1}{2}-\frac{c_\text{bnd}(n,i)}{c_\text{bnd}(n,i-1)}\right)^2\right).
\end{align}
So we have 
\begin{align}
P_{n,i}>  P_{n,i-1}^2f(n,i),
\end{align}
and
\begin{align}
P_n=P_{n,n}>  P_{n,n-1}^2f(n,n)>P_{n,n-2}^2 f(n,n-1)^2f(n)\cdots >\prod_{i=1}^{n-1}f(n,i)^{2^{n-i}}.
\end{align}
Because $f(n,i)<1$ for all $i$, the cumulative  product  $\prod_{i=1}^{n-1} [1-f(n,i)]^{2^{n-i}}$ decreases with $n$ monotonically. So we have 
\begin{align}
P_n>\lim_{n\rightarrow\infty}P_n>\lim_{n\rightarrow\infty}\prod_{i=1}^{n-1} [1-f(n,i)]^{2^{n-i}}>0.006. \label{eq:lst}
\end{align}
 As $P_n$ is lower bounded, Algorithm~\ref{alg:4} has average runtime of $\mathcal O(n^2)$. Because the average runtime of Algorithm~2 is always lower than Algorithm~\ref{alg:4}, it also has average runtime of $\mathcal O(n^2)$.

\section{Projection success probability for preparing arbitrary quantum states}\label{sec:4}
Here, we discuss the projection success probability $p_s, p_s'$. In Sec.~\ref{sec:perfect} we discuss  perfect preparation and prove Result.~4; in Sec.~\ref{sec:cut} we discuss cut-off preparation and prove Result.~5. 

\subsection{Perfect preparation}\label{sec:perfect}
With preparation method used in Result.~2 for positive vectors, $p_s$ can be directly determined by the success probability of projecting the value qubit to $|0\rangle$:
\begin{align}
p_s=\frac{\sum_{i=0}^{N-1}|v_i|^2}{\sum_{i=0}^{N-1}|v_i|^2+\sum_{i=0}^{N-1}|1-v_i|^2}\geqslant\frac{1}{N}\sum_{i=0}^{N-1}|v_i|^2.\label{eq:ps_ieq}
\end{align}
$p'_s$ in Result 3 for general complex vectors is more involved. From Eq.~(5) to Eq.~(6), the projection operator can be written as $M'=|+\rangle\langle+|^{\otimes 2}\otimes \mathbb{I}_{2^{n+1}}\otimes (\mathbb{I}_{2^{n}}\otimes|+\rangle\langle+|)^{\otimes 3}$. The quantum state before measurement is [Eq.~(5) in main text]:
\begin{align}
|\Psi'_0\rangle=\left(|00\rangle|\bm{\tilde v^{abcd}}\rangle-|01\rangle|\bm{\tilde v^{bacd}}\rangle+i|10\rangle|\bm{\tilde v^{cbad}}\rangle-i|11\rangle|\bm{\tilde v^{dbca}}\rangle\right),
\end{align}
and the state after the projection is [Eq.~(6) in main text]:
\begin{align}
|\Psi'_1\rangle&=M'|\Psi_0\rangle\notag\\
&=\frac{1}{2}|+\rangle^{\otimes2}(|\bm{v^{a}}\rangle-|\bm{v^{b}}\rangle+i|\bm{v^{c}}\rangle-i|\bm{v^{d}}\rangle)\otimes|\bm{v^{\text{uni}}}\rangle^{\otimes3}\notag\\
&=\frac{1}{2}|+\rangle^{\otimes2}\sum_{i=0}^{N-1}|n,i\rangle\left[(v_i^{(a)}-v_i^{(b)}+iv_i^{(c)}-iv_i^{(d)})|0\rangle-(v_i^{(a)}-v_i^{(b)}+iv_i^{(c)}-iv_i^{(d)})|1\rangle)\right]\otimes|\bm{v^{\text{uni}}}\rangle^{\otimes3}\notag\\
&=\frac{1}{2}|+\rangle^{\otimes2}\sum_{i=0}^{N-1}v_i|n,i\rangle\left(|0\rangle-|1\rangle\right)\otimes|\bm{v^{\text{uni}}}\rangle^{\otimes3}\notag\\
&=\frac{1}{\sqrt{2}}|+\rangle^{\otimes2}|\psi(\bm{v})\rangle\otimes|-\rangle\otimes|\bm{v^{\text{uni}}}\rangle^{\otimes3}.
\end{align}
Therefore, $p'_s$ can be calculated as
\begin{align}
p'_s=\frac{|\langle\Psi'_\text{1}|\Psi'_\text{1}\rangle|}{|\langle\Psi'_\text{0}|\Psi'_\text{0}\rangle|}.\label{eq:ps'_def}
\end{align}
Because $|\langle\bm{ v^{a}}|\bm{ v^{a}}\rangle|=\sum_{i=0}^{N-1}(v^a_i)^2+(1-v^a_i)^2\leqslant N $ and similar for $|\bm{ v^{b}}\rangle, |\bm{\tilde v^{c}}\rangle, |\bm{ v^{d}}\rangle$, we have $|\langle\bm{\tilde v^{abcd}}|\bm{ v^{abcd}}\rangle|\leqslant N^4$. So we have $|\langle\Psi'_\text{0}|\Psi'_\text{0}\rangle|\leqslant 4N^4$. Besides, one can calculate that $|\langle\Psi'_\text{1}|\Psi'_\text{1}\rangle|= \frac{N^3}{16}\sum_{i=0}^{N-1}|v_i|^2$. Therefore,
\begin{align}
p'_s\geqslant \frac{1}{64N}\sum_{i=0}^{N-1}|v_i|^2.\label{eq:ps'_ieq}
\end{align}
In the following, we prove Result.~4 for sampling case 1 and sampling case 2 separately.

\subsubsection{Sampling case 1}
In case 1, $|v_i|$ distributes uniformly in $[0,1]$, so we have 
\begin{subequations}
\begin{align}
&\text{mean}\left(\frac{1}{N}\sum_{i=0}^{N-1}|v_i|^2\right)=1/3,\\
&\text{mean}\left(\frac{1}{N}\sum_{i=0}^{N-1}|v_i|^2+|1-v_i|^2\right)=2/3, 
\end{align}
\end{subequations}
For $x,y>0$, we have
\begin{align}\label{eq:ps1}
\text{Pr}[p_s\geqslant x/y]&\geqslant  \text{Pr}\left[\frac{\frac{1}{N}\sum_{i=0}^{N-1}|v_i|^2}{   \frac{1}{N}\sum_{i=0}^{N-1}|v_i|^2+|1-v_i|^2  }\geqslant x/y\right]\notag\\
&\geqslant \text{Pr}\left[\frac{1}{N}\sum_{i=0}^{N-1}|v_i|^2\geqslant x \text{ and } \frac{1}{N}\sum_{i=0}^{N-1}|v_i|^2+|1-v_i|^2\leqslant y\right] \notag\\
&\geqslant 1-\text{Pr}\left[\frac{1}{N}\sum_{i=0}^{N-1}|v_i|^2< x\right]-\text{Pr}\left[\frac{1}{N}\sum_{i=0}^{N-1}|v_i|^2+|1-v_i|^2> y\right],
\end{align}
where we have used the relation $\text{Pr}[A \text{ and } B]=1-\text{Pr}[A \text{ and } \overline{B}]-\text{Pr}[\overline{A} \text{ and } B]-\text{Pr}[\overline{A} \text{ and } \overline{B}]=1-\text{Pr}[A \text{ and } \overline{B}]-\text{Pr}[\overline{A} ] \geqslant1-\text{Pr}[\overline{B}]-\text{Pr}[\overline{A}]$.

According to the Chernoff bound, for any $t_1,t_2>0$, we have
\begin{subequations}\label{eq:ch1}
\begin{align}
\text{Pr}\left[\frac{1}{N}\sum_{i=0}^{N-1}|v_i|^2< x\right]&\leqslant\left[\text{mean}(e^{-t_1|v_i|^2})\right]^Ne^{t_1Nx}=\left[\frac{\sqrt{\pi}\text{erf}(\sqrt{t_1})}{2\sqrt{t_1}}e^{t_1x}\right]^N\leqslant\left[\frac{\sqrt{\pi}}{2}\frac{e^{t_1x}}{\sqrt{t_1}}\right]^N,
\\
\text{Pr}\left[\frac{1}{N}\sum_{i=0}^{N-1}|v_i|^2+|1-v_i|^2> y\right]&\leqslant\text{mean}\left(e^{t_2(|v_i|^2+|1-v_i|^2)}\right)^Ne^{-Nt_2y}=\left[\frac{e^{t_2/2}\sqrt{\pi/2}\text{erfi}(\sqrt{t_2/2})}{\sqrt{t_2}}e^{-t_2y}\right]^N\leqslant\left[0.9\frac{e^{t_2(1-y)}}{\sqrt{t_2}}\right]^N.
\end{align}
\end{subequations}
Let $x=\frac{1}{5}(\delta/2)^{2/N}$, $y=1-\frac{1}{5}(\delta/2)^{2/N}$ and set $t_1=1/(2x)$, $t_2=1/(2-2y)$, Eq.~\eqref{eq:ch1} becomes
\begin{subequations}
\begin{align}
&\text{Pr}\left[\frac{1}{N}\sum_{i=0}^{N-1}|v_i|^2< \frac{1}{5}(\delta/2)^{2/N}\right]\leqslant\delta/2,\\
&\text{Pr}\left[\frac{1}{N}\sum_{i=0}^{N-1}|v_i|^2+|1-v_i|^2>-\frac{1}{5}(\delta/2)^{2/N} \right]\leqslant\delta/2,
\end{align}
\end{subequations}
So Eq.~\eqref{eq:ps1} becomes
\begin{align}
\text{Pr}\left[p_s\geqslant \frac{\frac{1}{5}(\delta/2)^{2/N}}{1-\frac{1}{5}(\delta/2)^{2/N}}\right]&\geqslant1-\delta.
\end{align}
In other words, $p_s$ is at the order of $\mit\Omega(\delta^{1/N})$. For $p_s'$, the process is similar.

\subsubsection{Sampling case 2}
Because $v_i=u_i/\max(|u_i|)$, we can rewrite Eq.~\eqref{eq:ps_ieq} and Eq.~\eqref{eq:ps'_ieq} as 
\begin{subequations}\label{eq:ps2_2}
\begin{align}
p_s&\geqslant\frac{1}{N}\frac{\sum_{i=0}^{N-1}|a_i|^2}{\max|a_i|^2}\label{ps_val},\\
p'_s&\geqslant \frac{1}{64N}\frac{\sum_{i=0}^{N-1}|a_i|^2}{\max{|a_i|^2}}\label{ps'_val}.
\end{align}
\end{subequations}
In case 2, $a_i=a_i^{(\text{r})}+ia_i^{(\text{m})}$, and the real and imaginary part of $a_i$ are independently sampled from standard normal distribution. $|a_i|$ follows Rayleigh distribution, and for $a_m\in(0,\infty)$, we have
\begin{align}
\text{Pr}[|a_i|<a_m]=1-e^{-a_m^2/2}.
\end{align}
So the maximum over all $|a_i|$ is given by 
\begin{align}
\text{Pr}[\max{|a_i|^2}<a^2_m]&=\left(1-e^{-a_m^2/2}\right)^N\notag\\
&\geqslant1-Ne^{-a_m^2/2}.
\end{align}
By setting $a_m=\sqrt{2\log (2N/\delta)}$, we have
\begin{align}
\text{Pr}[\max{|a_i|^2}\geqslant2\log (2N/\delta)]&<\delta/2.\label{eq:am}
\end{align}
Moreover, it can be calculated that
 \begin{subequations}
\begin{align}
\text{mean}\left(|a_i|^2\right)&=2,\\
\text{mean}\left(e^{-t{|a_i|}^2}\right)&=1/(1 + 2 t).
\end{align}
 \end{subequations}
According to the Chernoff bound, for any $t>0$, we have
\begin{align}
\text{Pr}\left[\frac{1}{N}\sum_{i=0}^{N-1}|a_i|^2\leqslant x\right]\leqslant \left(\frac{e^{tx}}{1+2t}\right)^N\leqslant\left(\frac{e^{tx}}{2t}\right)^N.
\end{align}
Let $x=\frac{1}{2}\left(\delta/2\right)^{1/N}$ and $t=1/x$, we have
\begin{align}
\text{Pr}\left[\frac{1}{N}\sum_{i=0}^{N-1}|a_i|^2\leqslant\frac{\left(\delta/2\right)^{1/N}}{2}\right]\leqslant \delta/2. \label{eq:ai_2}
\end{align}
Combining Eq.\eqref{eq:ps2_2}, Eq.\eqref{eq:am} and Eq.\eqref{eq:ai_2}, we have
\begin{align}
\text{Pr}\left[p_s\geqslant\frac{(\delta/2)^{1/N}}{4\log (2N/\delta)}\right]\geqslant&\text{Pr}\left[\frac{\sum_{i=0}^{N-1}|a_i|^2}{N\max{|a_i|^2}}
\geqslant\frac{(\delta/2)^{1/N}}{4\log (2N/\delta)}\right]\notag\\
\geqslant&1-\text{Pr}\left[\frac{1}{N}\sum_{i=0}^{N-1}|a_i|^2<\frac{(\delta/2)^{1/N}}{2}\right]- \text{Pr}\left[\max{|a_i|^2}>2\log 2N/\delta  \right]\notag\\
\geqslant& 1-\delta/2-\delta/2\notag\\
= &1-\delta.
\end{align}
In other words, $p_s=\mit\Omega\left(\frac{\delta^{1/N}}{\log (N/\delta)}\right)$. For $p_s'$ the process is similar.

\subsection{Cut-off preparation}\label{sec:cut}

 With the cut-off value $u_{\text{cut}}$ and define $\tilde{v}_i\equiv\arg(u_i)\min(|u_i|/u_\text{cut},1)$, the normalized cut-off target states becomes $|\phi(\bm{\tilde{v}})\rangle=\frac{1}{\sqrt{\sum_{i=0}^{N-1}|\tilde{v}_i|^2}}\sum_{i=0}^{N-1}\tilde{v}_i|n,i\rangle$. For sampling case 2, it can be further rewritten as 
\begin{align}
|\phi(\bm{\tilde{v}})\rangle=\frac{1}{\|\bm{\tilde{a}}\|_2}\sum_{i=0}^{N-1}\tilde{a}_i|n,i\rangle,
\end{align}
 where $\|\bm{\tilde{a}}\|_2\equiv\sqrt{\sum_{i=0}^{N-1}|\tilde{a}_i|^2}$, and $\tilde{a}_i\equiv\arg(a_i)\min(|a_i|,u_{\text{cut}}\|\bm{a}\|_2)$.
The proof of Result 5 follows from two lemmata as follow.

 \begin{lemma}\label{lm:2}
By setting $u_{\text{cut}}^2=\frac{8}{N}\left(4/\delta\right)^{1/N}\log\left(12/(\varepsilon_{\text{th}}\delta)\right)$, we have
\begin{align}
\text{Pr}\left[F\geqslant1-\varepsilon_{\text{th}}\right]
\geqslant&1-\delta/2.
\end{align}
\end{lemma}

\begin{proof}

Because
\begin{align}
|\langle\phi(\bm{u})|\phi(\bm{\tilde{v}})\rangle|=&\frac{1}{\|\bm{a}\|_2\|\bm{\tilde{a}}\|_2}  \sum_{i}  |a_i\tilde{a_i}|\notag\\
\geqslant&\frac{1}{{\|\bm{a}\|_2}^2}  \sum_{i}  |a_i\tilde{a_i}|\notag\\
=&\frac{1}{{\|\bm{a}\|_2}^2}\sum_{i} (|a_i|^2-|a_i|\max(0,|a_i|-u_\text{cut}\|\bm{a}\|_2))\notag\\
=&1-\frac{1}{{\|\bm{a}\|_2}^2}\sum_{i}|a_i|\max(0,|a_i|-u_\text{cut}\|\bm{a}\|_2),
\end{align}
we have
\begin{align}
F&=|\langle\phi(\bm{u})|\phi(\bm{\tilde{v}})\rangle|^2\notag\\
&\geqslant1-2\frac{1}{{\|\bm{a}\|_2}^2}\sum_{i}|a_i|\max(0,|a_i|-u_\text{cut}\|\bm{a}\|_2).
\end{align}
For any $x>0$, we have
\begin{align}
\text{Pr}\left[F\geqslant1-\varepsilon_{\text{th}}\right]\geqslant& \text{Pr}\left[\frac{1}{{\|\bm{a}\|_2}^2}\sum_{i}a_i\max(0,|a_i|-u_\text{cut}\|\bm{a}\|_2)\leqslant\frac{\varepsilon_{\text{th}}}{2}\right]\notag\\
\geqslant&\text{Pr}\left[{\|\bm{a}\|_2}^2\geqslant Nx^2 \text{ and } \sum_{i}a_i\max(0,|a_i|-u_\text{cut}\|\bm{a}\|_2)\leqslant\frac{\varepsilon_{\text{th}}Nx^2}{2}\right]\notag\\
\geqslant&\text{Pr}\left[{\|\bm{a}\|_2}^2\geqslant Nx^2 \text{ and } \sum_{i}a_i\max(0,|a_i|-u_\text{cut}\sqrt{N}x)\leqslant\frac{\varepsilon_{\text{th}}Nx^2}{2}\right]\notag\\
\geqslant&1-\text{Pr}\left[{\|\bm{a}\|_2}^2< Nx^2\right]-\text{Pr}\left[\sum_{i}a_i\max(0,|a_i|-u_{\text{cut}}\sqrt{N}x)> \frac{\varepsilon_{\text{th}}Nx^2}{2}  \right]\notag\\
=&1-\text{Pr}\left[{\|\bm{a}\|_2}^2< Nx^2\right]-\text{Pr}\left[\sum_i\Delta_i> \frac{\varepsilon_{\text{th}}Nx^2}{2}  \right],
\label{eq:fid_1}
\end{align}
 where we have defined $\Delta_i=|a_i|\max(0,|a_i|-u_{\text{cut}}\sqrt{N}x)$. After some calculation, we find that for any $t>0$, we have
\begin{align}
&\text{mean}(e^{-t{|a_i|}^2})=1/(1 + 2 t)
\end{align}
and
\begin{align}
& \text{mean}(\Delta_i)\leqslant 2e^{-N(u_{\text{cut}}x)^2/2}.
\end{align}
From Chernoff bound, we have
\begin{subequations}
\begin{align}
&\text{Pr}\left[{\|\bm{a}\|_2}^2< Nx^2\right]\leqslant \left[\frac{e^{tx^2}}{1 + 2 t}\right]^N,\\
\end{align}
\end{subequations}
Let $x^2=\frac{(\delta/4)^{1/N}}{2}$ and $t=(\delta/4)^{-1/N}$, we have

\begin{align}
\text{Pr}\left[{\|\bm{a}\|_2}^2<N\frac{(\delta/4)^{1/N}}{2} \right]\leqslant \left(\frac{e^{0.5}}{1+2(\delta/4)^{-1/N}}\right)^N\leqslant\left(\frac{e^{0.5}}{2}\right)^N     \left(\frac{1}{(\delta/4)^{-1/N}}\right)^N\leqslant\delta/4.\label{eq:as}
\end{align}
Let $u_{\text{cut}}^2=\frac{1}{x^2}\frac{4}{N}\log\left(\frac{12}{\varepsilon_{\text{th}}\delta}\right)=\frac{8}{N}(\frac{4}{\delta})^{1/N}\log\left(\frac{12}{\varepsilon_{\text{th}}\delta}\right)$, from Markov inequality, we have

\begin{align}
&\text{Pr}\left[\sum_i\Delta_i>N\frac{\varepsilon_{\text{th}}x^2}{2}  \right]\leqslant\frac{N\text{mean}(\Delta_i)}{N\varepsilon_{\text{th}}x^2/2}\leqslant\frac{4\exp[-\log(144/\varepsilon_{\text{th}}^2\delta^2)]}{\varepsilon_{\text{th}}x^2}  \leqslant \frac{4\varepsilon_{\text{th}}^2\delta^2}{144\varepsilon_{\text{th}}\frac{(\delta/4)^{1/N}}{2}}   \leqslant  \frac{\delta^2}{18(\delta/4)}   \leqslant  \delta/4
\label{eq:dl}
\end{align}

Combining Eq.~\eqref{eq:fid_1}, Eq.~\eqref{eq:as} and Eq.~\eqref{eq:dl}, Lemma.~\ref{lm:2} holds true.
\end{proof}

\begin{lemma}\label{lm:3}

By setting $u_{\text{cut}}^2=\frac{8}{N}\left(4/\delta\right)^{1/N}\log\left(12/(\varepsilon_{\text{th}}\delta)\right)$, we have

\begin{align}
\text{Pr}\left[p_s\geqslant C_p\right]\geqslant1-\delta/2
\label{eq:sup_ps_1}
\end{align}
where

\begin{align}
C_p&= \frac{(\delta/4)^{1/N}}{u_{\text{cut}}^212N\log\left((4/\delta)\right)}\\
&= \frac{(\delta/4)^{2/N}}{96\log\left(8/(\varepsilon_{\text{th}}\delta)\right)\log\left(4/\delta\right)}
\end{align}

\end{lemma}

\begin{proof}

According to Eq.~\eqref{eq:ps_ieq}, we have $p_s\geqslant\frac{1}{N}\sum_{i=0}^{N-1}\min\left(\frac{|a_i|^2}{(u_{\text{cut}}{\|\bm{a}\|_2})^2},1\right)$. So for any $y>0$,

\begin{align}
\text{Pr}\left[p_s\geqslant\frac{(\delta/4)^{1/N}}{3u_{\text{cut}}^2y^2}\right]\geqslant&\text{Pr}\left[\frac{1}{N}\sum_{i=0}^{N-1}\min\left(\frac{|a_i|^2}{(u_{\text{cut}}{\|\bm{a}\|_2})^2},1\right)\geqslant \frac{(\delta/4)^{1/N}}{3u_{\text{cut}}^2y^2}\right]\notag\\
\geqslant&1-\text{Pr}\left[ {\|\bm{a}\|_2}^2  > y^2\right]-\text{Pr}\left[ \frac{1}{N}\sum_{i=0}^{N-1}\min\left(\frac{|a_i|^2}{(u_{\text{cut}}y)^2},1\right)< \frac{(\delta/4)^{1/N}}{3u_{\text{cut}}^2y^2}\right],\notag\\
\geqslant&1-\text{Pr}\left[ {\|\bm{a}\|_2}^2  > y^2\right]-\text{Pr}\left[ \frac{1}{N}\sum_{i=0}^{N-1}\min\left(|a_i|^2,(u_{\text{cut}}y)^2\right)< \frac{(\delta/4)^{1/N}}{3}\right],\notag\\
\geqslant&1-\text{Pr}\left[ {\|\bm{a}\|_2}^2  > y^2\right]-\text{Pr}\left[ \sum_{i=0}^{N-1}\Delta'_i< \frac{N(\delta/4)^{1/N}}{3}\right],
\label{eq:p_s_cut}
\end{align}
where we have defined $\Delta_i'=\min\left(|a_i|^2,(u_{\text{cut}}y)^2\right)$. Here, $u_{\text{cut}}^2=\frac{8}{N}\left(4/\delta\right)^{1/N}\log\left(12/(\varepsilon_{\text{th}}\delta)\right)$ as set in Lemma.~\ref{lm:2}.
For $0<t_1<1/2$, we have
\begin{align}
\text{mean}(e^{t_1|a_i|^2})=1/(1-2t_1).
\end{align}
So according to Chernoff bound, we have
\begin{align}
&\text{Pr}\left[ {\|\bm{a}\|_2}^2> y^2\right]\leqslant e^{-t_1y^2}\left[     
\text{mean}(e^{t_1|a_i|^2})\right]^N\leqslant e^{-t_1y^2}/(1-2t_1)^N.
\end{align}
let $y^2=4N\log\left(8/\delta\right)$ and $t_1=1/(4N)$. 
It can be verified that
\begin{align}
&\text{Pr}\left[ {\|\bm{a}\|_2}^2>y^2\right]\leqslant \frac{e^{-t_14N\log\left((8/\delta)\right)}}{(1-2t_1)^N} \leqslant \delta/8\left[\frac{1}{1-1/(2N)}\right]^N\leqslant\delta/4
\label{eq:al}
\end{align}
Moreover, we have
\begin{subequations}
\begin{align}
&\text{mean}(e^{-t_2\Delta'_i})=\frac{1+2t_2e^{-(u_{\text{cut}}y)^2(1+2t_2)/2}}{1+2t_2},\\
&\text{mean}\left(\sum_{i=0}^{N-1}\Delta'_i\right)=N\left(2-2e^{-(u_{\text{cut}}y)^2/2}\right)>\frac{N(\delta/4)^{1/N}}{3}.
\end{align}
\end{subequations}
Therefore, according to the Chernoff bound,
\begin{align}
&\text{Pr}\left[ \sum_{i=0}^{N-1}\Delta'_i< \frac{N(\delta/4)^{1/N}}{3}\right]\leqslant\left[\frac{1+2t_2e^{-(u_{\text{cut}}y)^2(1+2t_2)/2}}{1+2t_2}\right]^N e^{t_2\frac{N(\delta/4)^{1/N}}{3}}.
\end{align}
Because $u_{\text{cut}}^2y^2=32(4/\delta)^{1/N}\log\left(12/(\varepsilon_{\text{th}}\delta)\right)\log\left(8/\delta\right)>32\log12\log8$
, when $t_2>1$, we have $2t_2e^{-(u_{\text{cut}}y)^2(1+2t_2)/2}<1$, and
\begin{align}
&\text{Pr}\left[ \sum_{i=0}^{N-1}\Delta'_i< \frac{N(\delta/4)^{1/N}}{3}\right]\leqslant \left[\frac{1+1}{2t_2}\right]^N e^{t_2\frac{N(\delta/4)^{1/N}}{3}}= \left(\frac{e^{\frac{t_2}{3}(\delta/4)^{1/N}}}{t_2}\right)^N.
\end{align}
Let $t_2=3(4/\delta)^{1/N}$, we have
\begin{align}
\text{Pr}\left[ \sum_{i=0}^{N-1}\Delta'_i< N(\delta/4)^{1/N}\right]\leqslant\left[\frac{e}{3}\frac{1}{(4/\delta)^{1/N}}\right]^N\leqslant\delta/4.\label{eq:dt'_ed}
\end{align}
Combining Eq.~\eqref{eq:p_s_cut}, Eq.~\eqref{eq:al}, Eq.~\eqref{eq:dt'_ed} and $y^2=4N\log(8/\delta)$, Lemma~\ref{lm:3} holds true.
\end{proof}

Result.~5 follows directly from Lemma~\ref{lm:2} and Lemma~\ref{lm:3}:
\begin{align}
\text{Pr}\left[p_s\geqslant C_p\text{ and } F\geqslant1-\varepsilon_{th}\right]&\geqslant 1-\text{Pr}\left[p_s< C_p\right]-\text{Pr}\left[F<1-\varepsilon_{th}\right]\notag\\
&\geqslant1-\delta.
\label{eq:pf}
\end{align}
Note that 
\begin{equation}
C_p=\Omega\left(\frac{\delta^{2/N}}{\log\delta^{-1}\log\left(\delta^{-1}\varepsilon_{\text{th}}^{-1}\right)}\right).
\end{equation}
For $p_s'$, the proof of Result.~5 is similar.

\section{trade-off with unitary preparation}\label{sec:sup_trade}

\begin{algorithm} 
\caption{: $f_{\text{trade-off}}(\bm{x},c_0,n_{\text{u}})$  }  
\label{alg:trade}  
\begin{algorithmic}[1]

\STATE \textbf{If} $\bm{x}$ is $2^{n_{\text{u}}}$-dimensional:

\STATE \quad prepare $c_0$ copies of $|\bm{x}\rangle$ with unitary preparation method \textit{in parallel}

\STATE \quad \textbf{Output} $|\bm{x}\rangle^{\otimes c_0}$

\STATE \textbf{Else}: 

\STATE \quad let $\bm{x^{(a)}}\oplus \bm{x^{(b)}}= \bm{x}$ 

\STATE \quad query $f_{\text{para}}(\bm{x^{(a)}},c_0)$ and $f_{\text{para}}(\bm{x^{(b)}},c_0)$ \textit{in parallel}, get return $|\bm{x^{(a)}}\rangle^{\otimes c_a}$  and $|\bm{x^{(b)}}\rangle^{\otimes c_b}$  

\STATE \quad define $c_{\min}=\min\{c_a,c_b\}$

\STATE \quad perform transformation $|\bm{x^{(a)}}\rangle\otimes|\bm{x^{(b)}}\rangle\rightarrow|\bm{x}\rangle$ for $c_{\min}$ times \textit{in parallel}, with $c$ trials success

\STATE \quad\textbf{If} $c=0$:

\STATE \quad\quad go to line 6

\STATE \quad\textbf{Else}:

\STATE \quad\quad \textbf{Output} $|\bm{x}\rangle^{\otimes c}$

\end{algorithmic} 
\end{algorithm} 

As described in Algorithm~\ref{alg:trade}, the trade-off preparation method is similar to Algorithm.~2 in the main text expect that at line $\#2$, unitary preparation is applied when the input state is $2^{n_{\text{u}}}$ dimensional. Algorithm.~2 can be considered as an extreme case of Algorithm~\ref{alg:trade} when $n_u=1$. For larger $n_u$, the circuit depth is higher, but average runtime is lower.

As an example, in Fig.~\ref{fig:nisq}, we consider the preparation of 16-qubit target state. We use case 2 sampling method, and the results are averaged over $100$ random states. When estimating circuit depth and runtime, all controlled-swap gates and control-control swap gates are decomposed into single qubit and CNOT gates. As can be seen, with more available ancillary qubits, we have both lower circuit depth (of single qubit and CNOT gates) and lower runtime.

 \begin{figure} [!t]
\includegraphics[width=0.5\columnwidth]{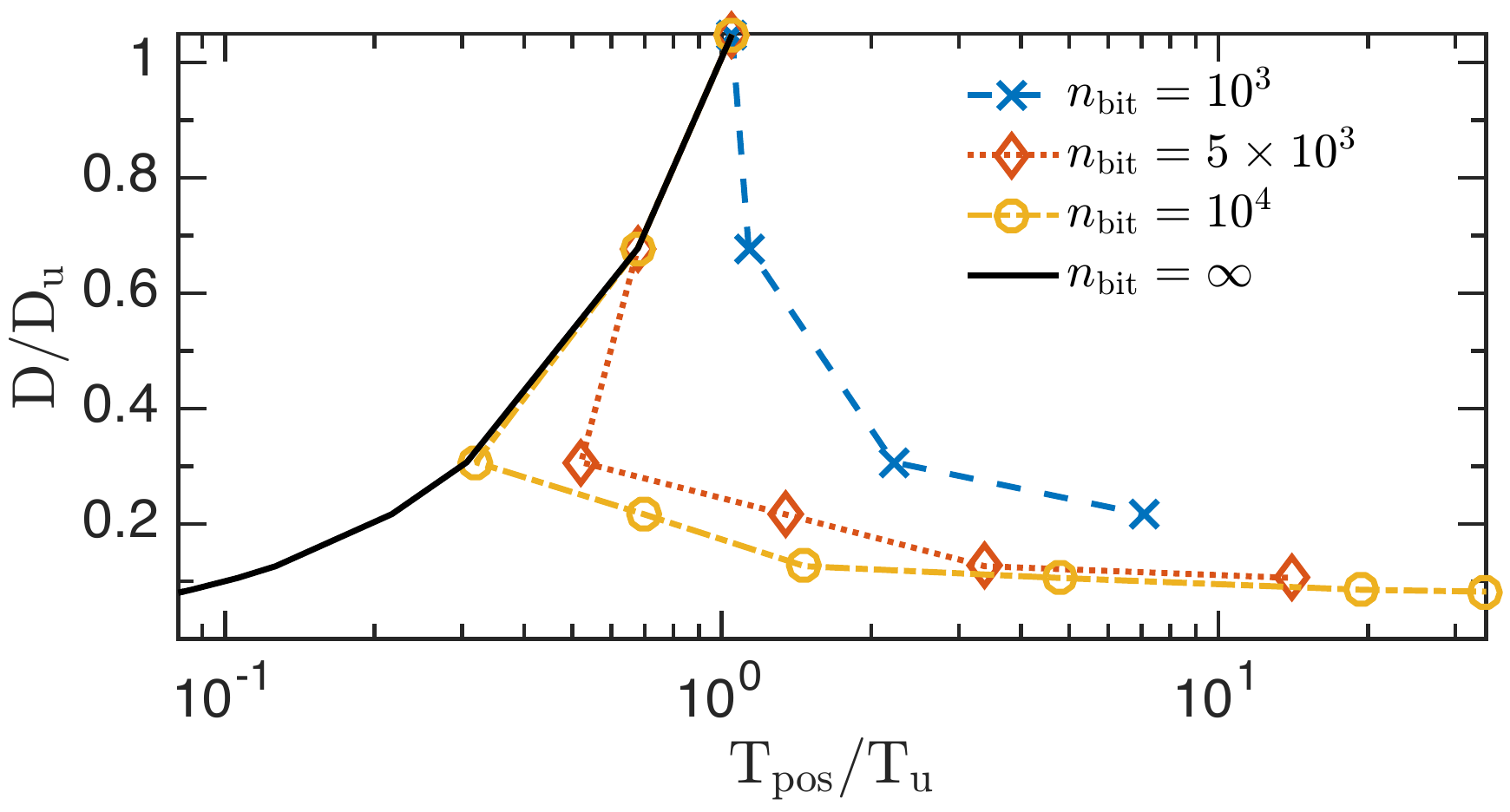}
\caption{Trade-off performance between parallel low-depth preparation and unitary preparation method. $n_\text{bit}$ is the total number of available qubits. $T_u$ and $D_u$ are the leading order of runtime and circuit depth of the unitary preparation method in Ref.~\cite{Plesch.11}. 
\label{fig:nisq}
}
\end{figure}

\newpage
\section{Lower bound of the circuit depth for preparing arbitrary quantum state}\label{sec:bnd}

Our protocols show that the state preparation can be done with $\mathcal O(n^2)$ circuit depth. It remains a question that if there is a lower bound of the circuit depth for state preparation and how close is our protocol to the fundamental limit. In this section, we will address this problem under the following framework:

\begin{enumerate}
 \item Initially, $\bm{v}$ is stored classically and there is no other prior knowledge about the target state;
	\item All operations are k-local; 
	\item Ancillary qubits are allowed, and qubits have all-to-all connections.

	\end{enumerate}

Under the framework above, we have the Result.~6.
The main idea of the proof for Result.~6 is as follow. Firstly, we introduce the concept of light cone, and show that the light cone size required for quantum state preparation is $\mit\Omega(N)$, as there are totally $N$ elements in $\bm{v}$. Secondly, we show that in order to obtain a light cone with $\mit\Omega(N)$, at least $\mit\Omega(n)$ layers of local operations are required. 

We begin with introducing several definitions.

\textit{Qubit connections and light cone.}
At each layer of the quantum circuit, we make grouping of all qubits. 
We denote $\pi_m^{(i)}$ as the $i$th group of qubits at the $m$th layer. Operations are applied only among the qubits in the same group. For example, if there are four qubits ($n_{\text{tot}}=4$), and $\pi_1^{(1)}=\{1,4\},\pi_1^{(2)}=\{2,3\}$, it means that at the first layer, qubits with label $1,4$ are connected to each other, and qubits with label $2,3$ are connected to each other. We also denote the group containing qubit $j$ at the $l$th layer as $\pi'_l(j)$. In the example above, we have $\pi'_1(1)=\{1,4\},\pi'_1(2)=\{2,3\},\pi'_1(3)=\{2,3\}$ and $\pi'_1(4)=\{1,4\}$. Note that at each layer, each qubit belongs only to \textbf{one} group. 

The light cone of qubit $j$ is defined as all qubits having connection to it. More rigorously, for a $L$-layer circuit, the light cone of qubit $j$ is
\begin{align}
&\mathcal S(j)\equiv \mathcal S_1(j),
\end{align}
where
\begin{eqnarray}
\mathcal S_m(j)\equiv \left\{
\begin{array}{ccl}
\pi'_m(j)      &    & m=L \\
 \bigcup_{k\in\mathcal{S}_{m+1}(j)}\pi'_m(k)    &  & 1\leqslant m<L   \label{eq:percep}\\  
\end{array} \right. \ .
\end{eqnarray}
Because the operations are k-local, the light cone size (number of elements of $S(j)$) satisfies $|S(j)|\leqslant k^L$. 

\textit{Quantum operations and classical encoding of $\bm{v}$.} Suppose there are totally $N_{\text{tot}}$ qubits ($N_{\text{tot}}$ is even), the operations at the $m$th layer generally takes the following form
\begin{align}
\mathcal{E}_m=\mathcal{E}_{\pi_m^{(N_{\text{tot}}/2)}}\circ\cdots\mathcal{E}_{\pi_m^{(2)}}\circ\mathcal{E}_{\pi_m^{(1)}}.
\end{align}
where $\mathcal{E}_{\pi_m^{(i)}}$ can be arbitrary quantum operations (with unitary, measurement and post-selection) in the subspace containing qubits in $\pi_m^{(i)}$. The full operation of a $L$-layer quantum circuit is just
\begin{align}
\mathcal{E}=\mathcal{E}_L\circ\cdots\circ\mathcal{E}_2\circ\mathcal{E}_1.\label{eq:el}
\end{align}

For simplicity, we restrict ourself to $\bm{v}\in[0,1]^N$, and the generalization to $\bm{v}\in\mathbb{C}$ is straight forward. As $\bm{v}$ has totally $N$ elements, it takes at least $N$ qubits (or classical bit) to store $\bm{v}$. Without loss of the generality, we assume that $\bm{v}$ is stored in the following form 
\begin{align}
\bigotimes_{i=0}^{N-1}\left(|v_i\rangle\langle v_i|\right)^{\otimes N_{\text{c}}},\label{eq:cls}
\end{align}
where we have allowed multiple copies ($N_{\text{c}}$) of the classical information. We also assume that the $n$-qubit system encoding $\bm{v}$ is initialized to $\rho_{\text{enc}}$, and there is an ancillary qubits system initialized to $\rho_{\text{anc}}$.

The target state takes the form of Eq.~(1) in the main text, so in the most general cases, quantum state preparation processes, $\mathcal E$, should satisfy the following
\begin{align}
\mathcal{E}\left(\rho_{\text{enc}}\otimes\rho_{\text{anc}}\otimes\bigotimes_{i=0}^{N-1}\left(|v_i\rangle\langle v_i|\right)^{\otimes N_{\text{c}}}\right)=\alpha(\bm{v})|\psi(\bm{v})\rangle\langle\psi(\bm{v})|\otimes\rho' \label{eq:prep_gen}
\end{align}
for all $\bm{v}\in[0,1]^N$. Here $\alpha(\bm{v})\in(0,1]$ as post-selections are allowed, and $\rho'$ can be an arbitrary quantum state of the joint system of ancillary qubits and the qubits encoding the classical information. 
Our goal is to find the lower bound of $L$ in Eq.~\eqref{eq:el} for quantum operations satisfying Eq.~\eqref{eq:prep_gen}.  

We first introduce a Lemma relating the light cone to the reduced density matrix (RDM) of the final output states. Lemma.~\ref{lm:lc} can be considered as a generalization of Lemma 2.3 in \cite{Barak.20}.

\begin{lemma}\label{lm:lc}
The RDM of qubit $j$ of the final output state  depends only on the RDM of the input state for the subsystem containing qubits in its light cone $\mathcal{S}(j)$.
\end{lemma}
\begin{proof}
The RDM of the final output state for qubit $j$ depends only on the RDM of $\mathcal S_L(j)$ at the $(L-1)$th layer.  

Similarly, the RDM of $\mathcal S_m(j)$ at the $m$th layer depends only on the RDM of $\mathcal S_{m-1}(j)$ at the $(m-1)$th layer. Therefore, the RDM of qubit $j$ of the final output state depends only on the RDM of total input state of $\mathcal S_1(j)=\mathcal S(j)$, i.e. the light cone of qubit $j$.

\end{proof}
The proof of Result.~6 follows directly from lemma.~\ref{lm:lc} as follow:

\begin{proof}
According to Eq.~\eqref{eq:prep_gen}, the final state of the encoding system and the RDM of each qubit in it depends on at least one copy of the state $\bigotimes_{i=0}^{N-1}|v_i\rangle\langle v_i|$.

According to Lemma.~\ref{lm:lc}, for qubits in the encoding system, the light cone $\mathcal S(j)$ contains at least $N$ qubits, i.e. $|\mathcal S(j)|\geqslant 2^n$. Suppose the elementary quantum operations are $k$-local, for $L$-layer quantum circuit, 
we have  $|S(j)|\leqslant k^L$. Therefore, in order to obtain an operation satisfying Eq.~\eqref{eq:prep_gen}, the circuit depth is lower bounded by $L\geqslant \frac{\log2}{\log k}n=O(n)$. 

\end{proof}


\begin{thebibliography}{49}
\bibitem [{not()}]{note_Nielsen}%
  \BibitemOpen
  \href@noop {} {}\bibinfo {note} {In pages 198-200, Section 4.5.4 of
  Ref\cite{Nielsen.02}, it is proven that with constant-weight quantum gates
  and no ancillary qubit, there are states of $n$ qubits which take
  $\mit\Omega(N/\log n)$ constant-weight operations to approximate to within a
  fixed distance (Eq.4.85).}\BibitemShut {Stop}%
\end{thebibliography}
\end{document}